%% file: articleDPMain_V3.tex
\documentclass[english, a4paper,11pt]{amsart}
\usepackage[utf8]{inputenc}
\usepackage[english]{babel}
\usepackage[foot]{amsaddr}
\usepackage{amssymb}
\usepackage{amsmath}
\usepackage{amsfonts}
\usepackage[vlined, ruled]{algorithm2e} 

\newtheorem{lemma}{Lemma}

\newtheorem{theorem}{Theorem}

\def\RR{{\mathbb{R}}}
\def\NN{{\mathbb{N}}}
\def\QQ{{\mathbb{Q}}}

\def\CC{{\mathcal{C}}}
\def\DD{{\mathcal{D}}}

\begin{document}
	
	\title[An exact dynamic programming algorithm for block sale]{An exact dynamic programming algorithm, lower and upper bounds, applied to the large block sale problem}

	\keywords{mixed integer non linear programming, non convex optimization, dynamic programming, optimal portfolio liquidation, large block sale, heuristics}
		
	\author{David Nizard$^{1,*}$, Nicolas Dupin$^1$ and Dominique Quadri$^1$}
	\thanks{$^1$Laboratoire Interdsicplinaire des Sciences du Numérique, Université Paris-Saclay, Gif-sur-Yvette, France}
	\thanks{$^*$ Corresponding author: \textbf{nizard@lri.fr}}

	\maketitle
	\subsection*{Abstract}
	\begin{small}
		\input{abstract_V3.tex}
	\end{small}
	
	\section{Introduction}
		\input{intro_V3.tex}
		
	\section{Problem statement and notations}\label{sectionDefPb}
			\input{defProblem_V3.tex}
	
	\section{Exact dynamic programming algorithm}\label{sectionDPExact}

\input{exactDP_V3.tex}
	
	\section{Lower bounds of the initial problem}\label{sectionLB}
		\input{lowerBounds_V3.tex}
	
	
	\section{Upper bound using monotony}\label{sectionUB}
		\input{upperBounds_V3.tex}
	
	\section{Numerical Experiments}\label{sectionNE}
		\input{introNE_V3.tex}
		\subsection{Experiment setup and penalty function calibration}\label{sectionData}
		\input{dataCalibration_V3.tex}
		\subsection{Numerical Experiments for small and medium instances}\label{sectionNESmall}
			\input{numExp_SmallMed_V3.tex}
		\subsection{Numerical Experiments for large instances}\label{sectionNELarge}
			\input{numExp_Large_V3.tex}

	\section{Conclusion and perspectives}\label{sectionCCL}
		\input{ccl_V3.tex}
		
%
%

	\appendix
	\section{g function calibration tables}\label{sectionCalib}
	\input{calibTables_VF.tex}
	
	\bibliographystyle{siam}
	\bibliography{bib_articleDP_V3}	
	
\end{document}

%% file: abstract_V3.tex
In this article, we address a class of non convex, integer, non linear mathematical programs using dynamic programming.
The mathematical program considered, whose properties are studied in this article, may be used to model the optimal liquidation problem of a single asset portfolio, held in a very large quantity, in a low volatility and perfect memory market, with few market participants. In this context, the Portfolio Manager's selling actions convey information to market participants, which in turn lower bid prices and further penalize the liquidation proceeds we attempt to maximize.\\
We show the problem can be solved exactly using Dynamic Programming (DP) in polynomial time. However, exact resolution is only efficient for small instances. For medium size and large instances, we introduce dedicated heuristics which provide thin admissible solutions, hence tight lower bounds for the initial problem. We also benchmark them against a commercial solver, such as LocalSolver \cite{localsolver_2020}.
We are also interested in the continuously relaxed problem, which is non convex. 
Firstly, we use continuous solutions, obtained by free solver NLopt \cite{john_2021} and transform them into thin admissible solutions of the discrete problem.  
Secondly, we provide, under some convexity assumptions, an upper bound for the continuous relaxation, and hence for the initial (integer) problem.
Numerical experiments confirm the quality of proposed heuristics (lower bounds), which often reach the optimal, or prove very tight, for small and medium size instances, with a very fast CPU time. Our upper bound, however, is not tight.

%% file: intro_V3.tex
In this article, we are interested in solving a non convex, integer, non linear mathematical program, in a maximization context, with a linear constraint.\\
Although the class of non convex mixed integer non linear problems (MINLP) is used to model a wide range of real world phenomena, it often comes with a steep price tag. These problems are generally NP-hard \cite{garey_1979}, which make them much harder to solve both in theory and practice, than convex non linear programs. It is so, because continuous relaxation of convex MINLP remains convex, which makes it, at least in theory, easier to handle. Indeed, in this convex context, a large panel of efficient methods were put forward early in the literature. For instance, Benders decomposition was introduced in \cite{geof_1972}, branch \& bound based approaches in \cite{gupta_1985,ques_1992}, and outer-approximation methods in \cite{duran_1986}. 
However, the particular subclass of convex quadratic pure integer programs has been covered, either by a straightforward branch \& bound, or by transforming the initial problem into a linear one. One can refer to \cite{bill_2012} and \cite{quad_2015} for examples of such approaches. In practice, current commercial solvers achieve excellent performances on this subclass.\\
Back to the general non convex case, where no general efficient algorithm is known, we find specific approaches in the literature, depending on the properties of the objective function $f$. Interestingly, the quadratic program, which has numerous applications in engineering and finance, has again been studied extensively (cf. survey \cite{flou_1995}).\\ 
More generally, when we assume $f$ differentiable over a compact set and attempt to minimize $f$, we know the minimum exists. A necessary condition for global minimum is stationarity $(\nabla f(\hat{x})=0)$. If furthermore, local convexity is achieved $(\nabla^2f(\hat{x})=0)$, then $\hat{x}$ is a local minimum. The remaining question is how to extend this local property to a global minimum ? In fact, as mentioned in survey \cite{hiri_1995}, a sufficient condition is that $f$ and its convex hull, denoted $cof$ coincide locally. Consequently, $\hat{x}$ is a global minimum if and only it is stationary and $cof(\hat{x})=f(\hat{x})$. In practice, we only shifted the problem because expressing $cof$ is a very arduous task. Nonetheless, a large body of literature is devoted to build a sequence of convex functions, inferior to $f$ on the admissible domain, referred to as convex under-estimators, which converge to the global minimum. Specifics of the algorithm and assumption on $f$ to ensure convergence are problem dependent. Reader can refer to \cite{horst_1995} and for a comprehensive presentation of current available techniques.\\     
In practice, quite logically, few solvers can cope with the general non convex case. 
Commercial solver CPLEX \cite{cplex_2018}, from its version 12.6, made progress in this direction by solving non convex quadratic programs with mixed integer variables, using spatial branch \& bound techniques. Recently, solvers such as BARON (distributed under GAMS \cite{gams_2021}) are able to solve this type of problems, but not for every objective function. For instance, BARON can not handle objective functions with an arctangent.\\
We also mention the free solver NLopt \cite{john_2021} which provides continuous solutions for a large spectrum of functions. While it implements different algorithms, we empirically found it works best, for our problem, using the gradient based optimization method developed in \cite{svan_2002}, which is based on conservative convex separable approximations and provide stationary points. In this paper, we only use NLopt in the context of our two-step approach, as discussed in section \ref{sectionNLO}.\\
In this study, we shifted our focus towards a specific mathematical program, whose properties we examine, and which originates from a real problem well known by financial practitioners: the optimal portfolio liquidation. We consider the case of a single asset portfolio, held in a very large quantity, in a low volatility and perfect memory market, with few market participants. In this context, the Portfolio Manager's selling actions convey information to market participants, which in turn lower bid prices and further penalize liquidation proceeds we attempt to maximize.\\
We show the problem can be solved exactly using Dynamic Programming (DP) in polynomial time. However, exact resolution is only efficient for small instances. For medium size and large instances, we introduce dedicated heuristics which provide thin admissible solutions, hence tight lower bounds for the initial problem. We also benchmark them against a commercial solver, such as LocalSolver \cite{localsolver_2020}.
We are also interested in the continuously relaxed problem, which is non convex. 
Firstly, we use continuous solutions, obtained by free solver NLopt \cite{john_2021} and transform them into thin admissible solutions of the discrete problem.  
Secondly, we provide, under some convexity assumptions, an upper bound for the continuous relaxation, and hence for the initial (integer) problem.
Numerical experiments confirm the quality of proposed heuristics (lower bounds), which often reach the optimal, or prove very tight, for small and medium size instances, with a very fast CPU time. Our upper bound, however, is not tight.\\
In the financial literature, the Large Block Sale (LBS) problem, which belongs to the larger class of optimal liquidation portfolio problems has been extensively studied through many different angles.
A first topic of interest was the market impact of large block trades. What is the performance of the asset subject to a LBS, under different time horizons ? 
How long does it take for the market to recover ? \cite{guth_1965, dann_1977,gemm_1996} provide empirical studies related to such questions.\\ 
A second alley studies specific markets dedicated to such transactions, known as \textit{upstairs market}, and their more recent electronic counterpart the \textit{dark pools}. The reader can refer, for instance to \cite{madh_1997} and \cite{mish_2017}.
A third approach was to model incentives of involved market participants in a block transaction (block trader, broker, specialist etc.),  study the price equilibrium, and derive existence and unicity conditions. \cite{sepp_1990} and \cite{keim_1996} provide examples.\\
Another type of model where price impact is inferred from the modeling of the Limit Order Book (LOB), which ranks best bid and ask prices and their corresponding volumes, was proposed in \cite{obiz_2005}. Refinements in that direction can be found in \cite{alfo_2008,alfo_2010} and \cite{obiz_2013}. \\
There are also abundant references for optimal liquidation in continuous time, where LBS turns into a stochastic control problem (\cite{barle_1994}, \cite{vath_2006}, \cite{seydel_2009}, \cite{khar_2010}, \cite{gath_2010,gath_2012}, just to mention a few).\\ 
However, in this article we consider LBS in discrete time, with integer variables. Our key question, how to optimally split the block into smaller orders in order to maximize the sale proceeds or equivalently minimize overall price impact, has been investigated in \cite{bert_1998}, \cite{almg_2000,almg_2003}.  \\
The corresponding models introduce a stock price dynamic, which accounts for previous trades size. The majority of the models distinguish between the \textit{temporary} price impact, due to the temporary imbalances between supply and demand at a given point in time, and the \textit{permanent} impact which reflects the effect on share price of the information conveyed by our trading (as defined in \cite{almg_2003}). Their resolution is often based on dynamic programming.\\
In most cases, related impact price functions obey linear or power laws. In addition, the majority of models feature either a \textit{short memory} in the sense their penalty at time $t_k$ is in $g(x_k)$, depending only of the $x_k$ units traded at that time, or their long term memory is separated (e.g \cite{almg_2000}), as the penalty at time $t_k$ is in $\sum_{j=1}^{k} g(x_j)$.\\
In this paper, as in \cite{boyd_2017}, we do not forecast the stock prices, which we assume to be known (or at least correctly estimated) prior to optimization. Therefore, we assume deterministic asset prices. Prices are actually simulated, using a classical geometric Brownian Motion \cite{hull_2002}, independently of the resolution of our mathematical program.\\ 
From a financial practitioner's standpoint, it is only realistic in very calm markets and for a low volatility asset, or very short time horizon (typically a day to a week), where no news or earnings are expected to be released. While it seems quite restrictive, it is in our experience, a very favorable environment to execute large block trades. On the contrary, in agitated markets, liquidity is scarce and/or large block sales of very volatile assets (which are not so frequent) are driven by prices evolution regardless of execution costs.
More importantly, this assumption insulates the market participants response to the liquidation, which constitutes our main object of study, from the evolution of asset prices. Consequently, our resolution algorithms and optimization techniques are valid regardless of the price vector. \\
Hence, our work aim to show how to best liquidate our single asset portfolio, provided asset prices evolution are correctly estimated, no matter how these estimations are made. We also included the case of constant prices during liquidation in numerical experiments, as a benchmark to best measure the dissemination of information in the market.\\
Outside of finance, the price vector $p$ can be interpreted as the \textit{standard} behavior of the environment unaltered by our actions. Reformulated in this context, we assume the standard behavior to be known, and we aim to study how our actions, which convey information to the environment with perfect memory, influence its response in order to maximize the output, or equivalently minimize our impact. \\

The contributions of our work are summarized as follows:
\begin{itemize}
	\item We propose a long term memory model, with a price impact function in $g(\sum_{j=1}^{k}x_j)$, for non linear bounded functions g, which is, to the best of our knowledge, new in the financial literature, related to the optimal liquidation problem. 
	\item We use dynamic programming to solve our non convex non linear integer program, exactly for small size instances. 
	\item We present a two-step method, based on an adapted dynamic programming algorithm, to compute heuristics for medium size and large instances. It yields very tight lower bounds, if not optimal solutions. We show this approach can be coupled with continuous optimization techniques to refine lower bounds of the initial problem. We also obtain an upper bound, while not tight, under some convexity assumptions on the objective.
	\item For almost every instance where our heuristics converge in tractable time, we beat the commercial solver LocalSolver v9.5 (cf. \cite{localsolver_2020}) and achieve a much higher optimal coverage ratio.
\end{itemize}

This article is organized as follows. Section \ref{sectionDefPb} defines the problem and its notations. Then, section \ref{sectionDPExact} introduces Bellman equation and solves the problem using dynamic programming. In section \ref{sectionLB}, we present different methods to obtain heuristics and discuss their complexity. In section \ref{sectionUB}, we study the continuously relaxed problem and compute an upper bound. We present our numerical results in section \ref{sectionNE}. Finally, section \ref{sectionCCL} concludes and presents futures perspectives for our work.

%% file: defProblem_V3.tex
We consider an investor in a financial market, who holds $N$ units of an asset and wishes to liquidate them over a given time horizon, split in $T$ time steps.\\
He decides to sell at each time step the quantity $x_t \in [\![0;N]\!] $ of a single asset. Hence, the admissible decisions set is defined by:
\begin{equation}
 \DD = \left\{ (x_1,\ldots,x_T)\in \NN^T, \phantom{1} \sum_{t=1}^{T} x_t = N \right\} 
\end{equation}

For each time step $t \in [\![1;T]\!]$, $x_t\geq 0$, we assume a sell only program, and let $p_t>0$ be the asset best bid price.
We also assume there is a minimal floor price $q_t>0$ for very large block trades, with $q_t<p_t$. To simplify notations, we define $c_t = p_t - q_t >0$.
We also introduce a strictly increasing function $g$ from $\NN$ to $[0,1[\;$, such that $g(0)=0$ and $\displaystyle\lim_{+ \infty} g = 1$.\\
At a given time $t$, the execution price of a block of $x_t$ units of asset is modeled by:
\begin{equation}
	v_t(x_t)= p_t\, x_t - (p_t-q_t)\cdot x_t \cdot g\left(\displaystyle\sum_{k=1}^{t} x_k\right) = \left[p_t- c_t\cdot g\Big(\displaystyle\sum_{k=1}^{t} x_k\Big)\right]x_t
\end{equation}
Let f be the objective value function from $\NN^T\rightarrow\RR$:
\begin{equation}
 f(x)=\displaystyle\sum_{t=1}^{T}v_t(x_t)=\displaystyle\sum_{t=1}^{T}\Big[p_t- c_t\cdot g\Big(\displaystyle\sum_{k=1}^{t} x_k\Big)\Big]x_t
\end{equation}
Penalty function $g$ models the market response to the investor selling action, taking into account market memory.
$p_t$ is indeed only valid for a very low traded volume compared to N. Higher volumes lead to lower execution prices, but the pace of decrease reflect market participants information about our intent. The more information, the lower the price. Penalty function $g$ is applied to $y_t=\sum_{k=1}^{t} x_k$, the asset's liquidated quantity up to current time t. $y_t$ represents the available information in a market with perfect memory. Hence, $q_t$ corresponds to the maximum impact, where the market is fully aware of our intentions and react accordingly. The penalty is increasing in $y_t$. The higher $y_t$, the closer the executed price to $q_t$, which justifies our assumptions for $g$.

Therefore, the investor's problem consists of maximizing the sale proceeds from his holding: 
\begin{equation}
	\mbox{OPT} = \displaystyle\max_{x \in \DD} f(x) \label{pbMain0}
\end{equation}

Problem (\ref{pbMain0}) is well defined, as $\DD$ is finite. In the most general case, (\ref{pbMain0}) is an non convex, non separable, non linear, integer optimization problem.\\
  
One will also be interested in the continuous relaxation of the problem. Hence, we extend previous notations.
Let $\CC$ be the set corresponding set of admissible solutions: 
\begin{equation}
	\CC = \left\{ (x_1,\ldots,x_T)\in \RR_+^T, \phantom{1} \sum_{t=1}^{T} x_t = N \right\}
\end{equation}
\begin{lemma}\label{lemCompactC}
	$\CC$ is a compact set.
\end{lemma} 
\begin{proof} Let $T$ be an integer (corresponding to the number of time steps), $N$  be a real number and $u$ be the function from $\RR^T\mapsto \RR$, such as $u(x)=\sum_{i=1}^T x_i$.\\
	Singleton set $\{N\}$ is a closed set of $\RR$ (because $\RR\backslash\{N\}$ is an open set). Hence its inverse image by the continuous function $u$ is a closed set of  $\RR^T$ (topological characterization of continuity).
	Moreover, $[0;N]^T$ is a closed bounded set of $\RR^T$.\\ 
	Therefore, $\CC = u^{-1}(\{N\})\cap[0;N]^T$ is a closed bounded set of $\RR^T$ which proves its compacity.\\ 
\end{proof}
Let $\overline{g}$ be the real extension of $g$, defined on $\RR_+$. It is strictly increasing  on $\RR_+$, and such that $\overline{g}(0)=0$ and $\displaystyle\lim_{+ \infty} \overline{g} = 1$.
Let $\overline{v_t}(x_t)=\left[p_t - c_t\cdot\overline{g}\Big(\displaystyle\sum_{k=1}^{t} x_k\Big)\right] x_t$ \\
Let $\overline{f}$, from $\RR^T\rightarrow \RR$, such that $\overline{f}(x)=\displaystyle\sum_{t=1}^{T} \overline{v_t}(x_t)$ \\
We finally consider the following optimization problem:\\

\begin{equation}\mbox{UB}_1 = \displaystyle\max_{x \in \CC}\overline{f}(x) \label{pbCont}
\end{equation}
In the remaining of the paper, we will assume $\overline{g}$ is continuously differentiable. Therefore, problem (\ref{pbCont}) is also well defined, as $\CC$ is a compact set, according to Lemma \ref{lemCompactC} and $\overline{f}$ is continuous. \\
As expected, problem (\ref{pbCont}) is in general non convex. So getting an upper bound of (\ref{pbMain0}), requires global solving, which often proves to be as hard as solving the initial problem. There are nonetheless several motivations to study problem (\ref{pbCont}).\\
Firstly, a global resolution of either the continuous relaxation, or any problem with a higher objective function, leads to an upper bound of the initial problem (\ref{pbMain0}). We present and solve such a problem in section \ref{sectionUB}, and denote its optimal value $\overline{\mbox{UB}_2}$, so that $\mbox{OPT}\leqslant \mbox{UB}_1\leqslant\overline{\mbox{UB}_2}$.\\    
Secondly, when $N \gg T$, discrete approximation is accurate enough, so that discrete and continuous modeling should be close. Indeed, we show continuous optimization techniques can be efficiently used to refine lower bounds of the initial problem (\ref{pbMain0}).\\
Therefore, our purpose is either to solve problem (\ref{pbMain0}) exactly, under a CPU time limit, or to provide an interval for the optimal value.\\
For numerical experiments, we selected concave functions with different convergence speeds to infinity: $g(x) = 1 - \frac 1 {1+x}$, $g(x) = 1 - \frac 2 {1+\sqrt{1+x}}$ ou $g(x) = \frac 2 {\pi} \mbox{arctan}(x)$.\\
Concavity is equivalent to decreasing first order derivative, which makes sense for numerical experiments.\\
While we selected a few penalty functions for numerical experiments, any strictly increasing positive function $h$, satisfying aforementioned regularity conditions is eligible. Indeed, we can define the corresponding function $g$ on $\RR_+$, by $g: x\mapsto \frac{h(x)-h(0)}{h(L_{\infty})-h(0)}$, for some real $L_{\infty}>N$, so that when $y_t$ goes to N, $g$ can be arbitrarily close to 1. Hence the pool of candidates for penalty function is quite large.\\
Lastly, we did not apply any filter to the admissibility set $\DD$  to exclude unrealistic solutions, from an economic standpoint. For instance, the \textit{fire sale} $(M,0,\cdots,0)$, which liquidates the whole block in one shot at the first time step, is admissible, while it very seldom is in practice. We are not concerned by lack of liquidity in the market and potential trading halt, due to exchange circuit-breakers, triggered upon a strong selling action at any time step. We assume optimal solutions to be practically feasible. In most cases, solutions stemming from instances considered in numerical experiments remain below the system limits for most blue chip company stock.  We did not however carry out any further analysis about it.

%% file: exactDP_V3.tex
In this section, we prove formally Bellman's equation and describe the dynamic programming algorithm \cite{bell_1957} used for exact resolution. We consider the following optimization problem: 
\begin{equation}(P_{t,n})\left\{
	\begin{array}{l}
		\displaystyle\max_{(x_1,\ldots,x_t)} f(x)=\displaystyle\sum_{i=1}^{t}\Big[p_i-c_i\,
		g\Big(\displaystyle\sum_{k=1}^{i} x_k\Big)\Big]x_i \\
		s.t: h(x)=\displaystyle\sum_{i=1}^{t}x_i-n = 0 \\
		\forall i, x_i\in \mathbb{N} \\
		\forall i, 0 \leq x_i\leq n \\
	\end{array}
	\right.\label{pbMain}
\end{equation}
$(P_{t,n})$ is equivalent to problem (\ref{pbMain0}), for $t=T$ and $n=N$. It is a discrete problem with bounded decision variables. Hence solutions exist for $t\leq n$. Let $O_{t,n}$ be its optimal value. 

\begin{theorem}\label{thBell}
	$\forall t,n,\;1\leq t\leq n,\;O_{t,n}=\displaystyle\max_{i \in [\![0;n]\!]} \left\{ O_{t-1,n-i}  + \big[p_t-c_t\,g\big(n\big)\big]i\;\right\}$\\
	\hbox{with initial conditions} $O_{t,0}=O_{0,n}=0$. 
\end{theorem}
\begin{proof}
	By induction on t. \\
	Let n be fixed and $(t<n)$. For $t=1$, we tautologically sell the whole block at the only allowed time t. 
	\begin{equation*}
		\begin{split}
			O_{1,n}&=f(n) \\
			& = \big[p_1-c_1\,g\big(n\big)\big]\,n \\
			& = \displaystyle\max_{i \in [\![0;n]\!]}\left\{ \big[p_1-c_1\, g(n)\big]i\right\},\;since\;(p_1-c_1\,g(n))>0\\
			& = \displaystyle\max_{i \in [\![0;n]\!]}\left\{O_{0,n-i}  + \big[p_1-c_1\, g(n)\big]i\right\},\;since\;O_{0,n-i}=0\; \forall i\\
		\end{split}
	\end{equation*}
	which proves the result for t=1.\\
	Let us suppose the result be true for a given t.\\
	Let $g_i(x)\equiv\Big[p_i-c_i\,g\Big(\displaystyle\sum_{k=1}^{i} x_k\Big)\Big]$
	\\
	Let $(\widehat{x_1},\cdots,\widehat{x}_{t+1})$ be an optimal solution of $(P_{t+1,n})$. By definition, \\
	\[O_{t+1,n}=\displaystyle\sum_{i=1}^{t}g_i(\widehat{x})\,\widehat{x_i} + \big[p_{t+1}-c_{t+1}\,g(n)\big]\,\widehat{x}_{t+1} \]
	Since $\displaystyle\sum_{i=1}^{t}\widehat{x_i}=m -\widehat{x}_{t+1}$, the first t coordinates of $\widehat{x}$ constitute an admissible solution of $(P_{t,n-\widehat{x}_{t+1}})$. By our induction assumption,  its value function is therefore bounded by $O_{t,n-\widehat{x}_{t+1}}$. It yields\\
	\begin{equation*}
		\begin{split}
	O_{t+1,n}&\leq O_{t,n-\widehat{x}_{t+1}} + \big[p_{t+1}-c_{t+1}\,g(n)\big]\,\widehat{x}_{t+1}\\
	&\leq\displaystyle\max_{i \in [\![0;n]\!]} \left\{ O_{t,n-i}  + \big[p_{t+1}-c_{t+1}\,g(n)\big]i\;\right\}
		\end{split}
	\end{equation*}
	Conversely, for all $0\leq i\leq n$, $(P_{t,n-i})$ admits a solution. Hence,\\
	\[
		\exists(x_1,\cdots,x_t)\;s.t \left\{
		\begin{array}{l}
			O_{t,n-i}=\displaystyle\sum_{j=1}^{t}g_j(x)\,x_j \\
			\displaystyle\sum_{j=1}^{t}x_j=n-i	
		\end{array}
		\right.
	\]
	The extension of x by i $(x_1,\cdots,x_t,\mathbf{i})$ is admissible for $(P_{t,n})$. Hence, its value function is bounded by $O_{t+1,n}:$\\
	\begin{align*}
		f(x) &\leq O_{t+1,n}\\
		\displaystyle\sum_{j=1}^{t}g_j(x)\,x_j + \big[p_{t+1}-c_{t+1}\,g(n)\big]\,i & \leq O_{t+1,n}\\
		O_{t,n-i} + \big[p_{t+1}-c_{t+1}\,g(n)\big]\,i & \leq O_{t+1,n}
	\end{align*}
	Since this inequality can be established for any i in $[\![0;n]\!]$, we take the maximum on i and conclude:\\
	\[
		\displaystyle\max_{i \in [\![0;n]\!]} \left\{O_{t,n-i} + \big[p_{t+1}-c_{t+1}\,g(n)\big]\,i\right\} \leq O_{t+1,n}
	\]
	 which proves the results for $(P_{t+1,n})$.
\end{proof}
Bellman equation from Theorem \ref{thBell} provides an explicit scheme to solve problem (\ref{pbMain0}), or equivalently $(P_{T,N})$, through dynamic programming. We present the exact resolution in the next paragraph.\\

We first build the T by N matrix $\big(O_{t,n}\big)_{t\leq T,\, n\leq N}$ which contains optimal value of the sub-problems. Then, we derive the optimal strategy by backtracking in Bellman equation.  We present the exact DP algorithm, referred to as Algorithm \ref{algoDPExact} :\\\\
\begin{algorithm}[H]
	\DontPrintSemicolon
	\textit{\textbackslash\textbackslash Build $\big(O_{t,n}\big)$ matrix}\;
	\For{$t=1$ \KwTo $T$}{
		\For{$n=1$ \KwTo $N$}{
			\For{$k=0$ \KwTo $n$}{
				$O_{t,n}= \max\Big(O_{t,n},O_{t-1,n-k}+ \big[p_t-c_t\,g\big(n\big)\big]\,k\;\Big)$\; 
				}
			}
	}
	\textit{\textbackslash\textbackslash compute optimal strategy by backtracking}\;
		$n= N$\;
	\For{$i=T$ \KwTo 1}{
		\eIf{$O_{i,n}=O_{i-1,n}$}{ 
			$x_{i} = 0$\;
		}{
			$k= n$\;
			\While{\Big($O_{i,n}\neq O_{t-1,n-k}+\big[p_t-c_t\,g\big(n\big)\big]\,k\Big)\;\land\;(k>0)$}{
				$k= k-1$\;
			}
			$x_i= k$\;
			$n= n-k$\;
		}
	}	

	\label{algoDPExact}
	\caption{Exact Dynamic Programming}
\end{algorithm}
Space (computer memory) complexity is $O(TN)$. Moreover, in the third loop, we have to perform N evaluations, at worst, to determine $O_{t,n}$. Hence, it takes $O(TN^2)$ to compute the whole T by N matrix. In the backtracking algorithm, for each $i\leq T$, we perform at worst N comparisons, so complexity of the backtracking is $O(TN)$. Therefore time complexity of the DP algorithm is $O(TN^2)$. \\
If we store $\big(g(0),\cdots,g(N)\big)$, data consist in vector $p,c,\;and \;g$. Hence data are $O(T+N)$. Hence time complexity is polynomial, cubic at worst since $O(TN^2)<O\big((T+N)^3\big)$.\\	
If the vector g(i) is not stored, data are $O(T)$. Therefore space complexity is pseudo polynomial (polynomial only when N is fixed).\\
However, in practice, exact resolution proves too costly in space/time for large instances. In the next paragraph, we provide different heuristic methods to get tight lower bounds.

%% file: lowerBounds_V3.tex
In this section, we present different methods to rapidly compute tight lower bounds. We start with naive heuristics, then we introduce two DP based algorithms and an Iterated Local Search (ILS) algorithm.
\subsection{Naive heuristics}
There are many ways to generate a set of positive integer of a given sum N. It could be generated randomly. In this subsection, we select two intuitive heuristics, which we will use as benchmark for calibration and numerical experiments.
\begin{itemize}
	\item The first one is the fire sale: $x=\left(M,0,\cdots,0\right)$: we liquidate the block in one shot at the first time step. We get logically the maximum penalty for this panic move. It almost never is the best strategy (for the asset class we consider) and in the constant prices case, it turns out to be actually the worst admissible strategy.
	\item The second one is the uniform sale strategy $x=\left(N/T,\cdots,N/T\right)$, which liquidates the block linearly with time. 
\end{itemize}

\subsection{Two-step DP based methods}\label{sectionCGrain}
This technique consists in two independent steps. We firstly derive an approximate but very fast heuristic, by applying DP to P-sized buckets. Secondly, we refine it by intensifying the research in its neighborhood using an adapted DP method, presented below in Algorithm \ref{algoDPBounds}.

\begin{itemize} 
	\item \textbf{Coarse grain DP}: when N is large, exact DP is too costly as time complexity grows as $T\,N^2$. So a natural idea consists of selling buckets of P units $(P\gg 1)$, which we refer to as grain P. Resolution algorithm is almost identical to the exact one with $N'=\lfloor N/P\rfloor$. Hence its time complexity is $O\left(\dfrac{T\,N^2}{P^2}\right)$, faster by a factor $P^2$. The coarser the grain the faster the heuristic, and the lesser its quality (distance to optimal). Hence, as often, there is a trade-off in grain between CPU time and heuristic quality, which we touch upon in section \ref{sectionNETSCoarse}. \\
	\item \textbf{DP with bounds}: assuming the previous stage heuristic is \textit{close} to the optimal, we apply the exact DP algorithm in its neighborhood. We restrict the search by imposing bounds on admissible solutions. We first introduce the following definitions.\\
	Let $x^0=\left(x_1^0,\cdots,x_T^0\right)$ be the initial solution.\\
	Let $l$ and $u$ the lower (resp. upper) bounds for x such that: $\forall\, t\; 0 \leq l_t\leq x_t\leq u_t\leq N$\\
	Let: $L_t= \min\left(\displaystyle\sum_{i=1}^{t} l_i, N\right)$ and $U_t=\min\left(\displaystyle\sum_{i=1}^{t}u_i,N\right)$.\\
	Hence, $\forall\,t\; L_t\leq y_t\leq U_t$, where $y_t=\displaystyle\sum_{i=1}^{t}x_i\\$\\ 
	We then rewrite the restricted Bellman equation from Theorem \ref{thBell}:\\  
	$O_{t,n}= \displaystyle\max_{l_t \leq k\leq u_t}\Big(O_{t-1,n-k}+ \big[p_t-c_t\,g\big(n\big)\big]\,k\;\Big)$, subject to $L_t \leq n=y_t\leq U_t$\\
	We introduce the DP with bounds algorithm, referred to as Algorithm \ref{algoDPBounds}: \\
\end{itemize}

\begin{algorithm}[H]
	\DontPrintSemicolon
	\textit{\textbackslash\textbackslash Build $\big(O_{t,n}\big)$ matrix}\;
	\For{$t=1$ \KwTo $T$}{
		\For{$n=L_t$ \KwTo $U_t$}{
			\For{$k=l_t$ \KwTo $u_t$}{
				$O_{t,n}= \max\Big(O_{t,n},O_{t-1,n-k}+ \big[p_t-c_t\,g\big(n\big)\big]\,k\;\Big)$\; 
			}
		}
	}
	\textit{\textbackslash\textbackslash backtracking is identical to Algorithm (\ref{algoDPExact})}\;
	\textit{\textbackslash\textbackslash Get solution by backtracking in $O_{T,N}$ matrix}\;
	
	\label{algoDPBounds}
	\caption{Dynamic Programming with bounds}
\end{algorithm}
Straightforward reading from previous algorithm yields time and space complexity. In the memory space, we need to store $O_{t,n}$, when n varies from $U_t$ to $L_t$, and t from 1 to T. Hence space complexity is $O\left(\sum_{t=1}^T U_t-L_t\right)$. In time, each $O_{t,n}$ requires $u_t-l_t$ computations, hence time complexity is $O\left(\sum_{t=1}^T (U_t-L_t)(u_t-l_t)\right)$.\\
Let us suppose $u_t-l_t$ is bounded by some $R\geq 0$ for all t. This is typically the case for the funnel around a given heuristic. We note $R$ can always be defined as decision variables are finite and bounded and $R\leq N$.\\
\begin{align*}
	U_t-L_t & = \sum_{i=1}^t u_i-l_i\\
	 & \leq t\,R\\
\mbox{and, } \sum_{t=1}^T (U_t-L_t)(u_t-l_t) & \leq R\,\sum_{t=1}^T (U_t-L_t)\\
	&\leq\sum_{t=1}^T t\,R^2\\
	&\leq\frac{T(T+1)R^2}{2}
\end{align*} 
Conclusion, when we bind our research neighborhood by $R$, space complexity becomes $O\left(T^2\,R\right)$ and time complexity is $O\left(T^2\,R^2\right)$.
\begin{itemize}
	\item These steps can work in synch. We firstly compute an heuristic based on a P grain using DP. Then, we refine the solution using DP with bounds, with a funnel of size $\lambda P$ around the first stage heuristic ($\lambda$ is a scalar). In this two-step approach, P and  $\lambda$ are the only degree of freedom. Selection of these parameters for best heuristics quality and time is an interesting question, which is discussed in section \ref{sectionNE}.
	\item The two steps can also be run independently. Any heuristic can be coupled with the DP with bounds algorithm. However, to improve the first stage heuristic, one has to suspect a good local maximum lies in the neighborhood. Typically a small funnel around the fire sale strategy will not yield a good result.
	\item Another interesting point is that the DP with bounds technique can be applied to a continuous solution $x^0$. For instance, $l_t=\max\left(0,\lfloor x_t^0\rfloor-\lambda P\right)$ and $u_t=\min\left(N,\lceil x_t^0\rceil+\lambda P\right)$ define admissible bounds. Hence, we can apply any continuous optimization technique to compute a good admissible continuous solution of problem \ref{pbCont} and then search discrete local maxima around it. Performances of this hybrid two-step method (continuous relaxation coupled with DP with bounds) are discussed in numerical experiments in section \ref{sectionNE}.
		\item On the contrary to exact DP, we can not guarantee results are optimal. It is the main drawback of the two-step method. 
\end{itemize}

 \subsection{Iterated Local Search (ILS)}\label{sectionILS}
In this section, we introduce an ILS algorithm, referred to as Algorithm \ref{algoILS}, which starts from an admissible solution and provides an admissible local maximum. Let $x^0$ be an admissible solution. We first shift $x_1^0$ by $+P$ and $x_2^0$ by $-P$, for some fixed integer $P$. We apply these shifts only if resulting decision variables remain in the feasible $[\![0;N]\!]$ domain. We can easily see that $x^1_1+x^1_2+\sum_{i=3}^T x^0_i=N$ hence a solution's admissibility is stable by the P-shift operator, for all P and all pair $(x_i,x_{i+1})$.\\
If this shift improves the value function, we store the gap between shifted and original solution. We then apply the opposite shift $-P$, subject to the same boundary conditions. Now we make i vary between 2 and T and proceed similarly with pairs $(x_i^0, \,x_{i+1}^0)$.\\
Consequently, we considered $2\,(N-1)$ potential shifts. If no transformation improves the objective value function the algorithm stops and returns $x^0$. Otherwise, we select the biggest gap, store the corresponding solution $x^1$ and iterate the same steps starting from $x^1$.\\
We iterate the process and return either a local maximum, which corresponds to a fixed point $x^{K+1}=x^K$ for the P-shift operator, or the best solution achieved under a preset time and number of iterations limit.\\
Lastly, to improve this local maximum, we can keep applying different $P'$-shifts, with $(P'\neq P)$. For instance, in numerical experiments, we proceed by dichotomy on P: $P_0=2^R,\,P_1=2^{R-1},\cdots, P_R=1$, and $R=\lfloor\frac{\ln(N)}{\ln(2)}\rfloor$ the largest integer where a $2^R$ shift may be possible.\\
This local search algorithm triggers a call to the objective value function f for every shift. We improve its efficiency by comparing the relevant terms of f, as described in the following Algorithm \ref{algoILS} :

\begin{algorithm}[H]
	\DontPrintSemicolon
	$runSum= 0$\;
	$auxSum= 0$\;
	$\Delta= 0$\;	
	$f_{max}= 0$\;	
	$i_{opt}= -1$\;
	
	\For{$i=1$ \KwTo N-1}{
		$d_+= 0$\;
		$d_-= 0$\;
		$runSum= runSum+x_i$\;
		$auxSum= runSum+x_{i+1}$\;
		$d_0=\big[p_i-c_i\,g(runSum)\big]\,x_i +\big[p_{i+1}-c_{i+1}\,g(auxSum)\big]\,x_{i+1}$\;
		\If{($x_i+P\leq M) \land (x_{i+1}-P\geq 0)$}{
			$d_+=\big[p_i-c_i\,g(runSum+P)\big]\,(x_i+P)+\big[p_{i+1}-c_{i+1}\,g(auxSum)\big]\,(x_{i+1}-P)$\;
		}			
		\If{($x_{i+1}+P\leq M) \land (x_{i}-P\geq 0)$}{
			$d_-=\big[p_i-c_i\,g(runSum-P)\big]\,(x_i-P)+\big[p_{i+1}-c_{i+1}\,g(auxSum)\big]\,(x_{i+1}+P)$\;
		}			
		\eIf{$(d_+\geq d_-)$}{ 
			\If{$(d_+-d_0 > \Delta)$}{
				$\Delta = d_+-d_0$\;
				$i_{opt}= i$\;
				$\epsilon= 1$\;
			}
		}{
			\If{$(d_--d_0 >\Delta)$}{	 
				$\Delta = d_--d_0$\;
				$i_{opt}= i$\;
				$\epsilon= -1$\;
			}
		}
		\If{$i_{opt}>-1$}{ 
			$x_{i_{opt}}= x_{i_{opt}} + \epsilon\cdot P$\;
			$x_{i_{opt}+1}= x_{i_{opt}} - \epsilon\cdot P$\;
			$f_{max}= f_{max} + \Delta$\;
		}
	}
	\label{algoILS}
	\caption{Iterated Local Search}
\end{algorithm}

We notice that, for every i, we increment the running sums $\displaystyle\sum_{j=1}^{i}x_j$ and $\displaystyle\sum_{j=1}^{i+1}x_j$ and compare only the impacted terms of the objective value function f, namely $d_0$ (no shift), $d_+$ (shift $(i,i+1,P)$ ) and $d_-$ (shift $(i,i+1,-P)$ ). \\ 
Regarding its application, it can be used, either as a first step heuristic, starting from a naive solution (e.g uniform sale), or as a second step to improve an existing local maximum.

 \subsection{Commercial discrete solver: LocalSolver}
Lastly, we use a commercial solver dedicated to local optima to benchmark against our lower bounds. We selected LocalSolver \cite{localsolver_2020} as one of the leading solvers regarding local search algorithms. We did not carry out a exhaustive comparison against every solver which solves at least locally problem (\ref{pbMain0}), as we believe LocalSolver results to be representative of the state of the art of commercial solvers relevant for this problem.

\subsection{Free continuous solver NLopt}\label{sectionNLO}
As mentioned in section \ref{sectionCGrain}, we can use continuous solution in the first stage of our two-step approach. In this paper, we compute admissible continuous solutions using the free/open-source solver NLopt \cite{john_2021}, which specializes in continuous nonlinear optimization. We tested every available gradient based algorithm relevant for our problem and selected the Conservative Convex Separable Approximation \cite{svan_2002}, in its quadratic version (CCSAQ), which empirically works best for the relaxed problem \ref{pbCont}. In a maximization context, this algorithm generates and solves concave separable subproblems using approximate objective and constraint functions at each iteration. Subproblems approximate functions, are deemed \textit{conservative} when they become inferior to the objective function and underlying constraints. It is globally convergent in the sense it converges toward the set of points satisfying Karush-Kuhn-Tucker (KKT) conditions, which is non empty since $\mbox{UB}_1$ exists. Again, without further convexity assumption, it does not provide the global maximum $\mbox{UB}_1$. In fact, because problem (\ref{pbCont}) is not convex, we found a few instances, where NLOpt solution was lower than the optimal discrete solution of problem (\ref{pbMain0}), although by a tiny margin.\\ We thus consider NLOpt heuristics for what they are, fairly good admissible solutions of the continuously relaxed problem (\ref{pbCont}) and hence prime candidate for neighborhood search during the second stage.

%% file: upperBounds_V3.tex
In the previous section, we proposed lower bounds of the discrete problem (\ref{pbMain0}). To obtain an interval for the solution, when the optimal is unknown, one needs a true upper bound, which is the object of this section.\\
We first notice that penalty function $g$ is assumed strictly increasing, and decision variable $x_i$'s are non negative: 
\begin{align*}
	\overline{g}\Big(\displaystyle\sum_{k=1}^{t} x_k\Big)\geq& \;\overline{g}(x_t)\\ \\
	-\overline{g}\Big(\displaystyle\sum_{k=1}^{t} x_k\Big)\leq&-\overline{g}(x_t)\\ 
	f(x)\leq&\displaystyle\sum_{t=1}^{T}\Big[p_t- c_t\cdot \overline{g}(x_t)\Big]x_t\\
\end{align*}
We can hence define the continuous optimization problem:
\begin{equation}
	\mbox{UB}_2 = \displaystyle\max_{x \in \CC}\,U(x) \label{pbSep}
\end{equation}
where
\begin{equation*}
	U(x)=\displaystyle\sum_{t=1}^{T}u_t(x_t)=\displaystyle\sum_{t=1}^{T}\Big[p_t- c_t\cdot \overline{g}(x_t)\Big]x_t
\end{equation*} 
Problem (\ref{pbSep}) is well defined, and $\mbox{UB}_2$ provides an upper bound of problem (\ref{pbMain0}) and (\ref{pbCont}). We also notice its value function is separated (it can be written as a sum of univariate functions).\\
We now introduce sufficient condition on $g$ to ensure problem (\ref{pbSep}) is concave.

\begin{lemma}\label{lemXG}
	Let function $x\mapsto x\,\overline{g}(x)$, defined in $\RR_+$, be strictly convex. Then function $U$ is strictly concave and problem (\ref{pbSep}) is concave.
\end{lemma} 
\begin{proof}
	It is straightforward from $u_t$ definition:
	\begin{equation*}
		u_t(x)=p_t\,x-c_t\,[x\,\overline{g}(x)]
	\end{equation*}
For all $t$, $c_t\geq 0$, hence  $-c_t\,[x\,\overline{g}]$ and consequently $u_t$ are strictly concave as the sum of a concave and linear functions.
Therefore $U$ is also concave as the sum of concave functions.\\
\end{proof}	
\begin{lemma}Concave functions selected for numerical experiments (cf. section \ref{sectionDefPb}): $g(x) = 1 - \frac 1 {1+x}$, $g(x) = 1 - \frac 2 {1+\sqrt{1+x}}$ or $g(x) = \frac 2 {\pi} \mbox{arctan}(x)$ satisfy the assumptions of lemma \ref{lemXG}.\\
\end{lemma} 
\begin{proof} By straightforward computation of $\big[x\,\overline{g}(x)\big]''$
\end{proof}
\subsection{Resolution of the separated problem}  
We compute the Lagrangian function (\ref{pbSep}) and solve it for stationary points. Constraint qualification condition is satisfied by linearity of the constraint, and problem \ref{pbSep} is convex. Therefore, resolution of Lagrange equations provides a global maximum.
\begin{align*}
	L(x,\lambda)=&\displaystyle\sum_{t=1}^{T}\Big[p_t- c_t\cdot \overline{g}(x_t)\Big]x_t - \lambda\Big(\displaystyle\sum_{t=1}^{T}x_t - N\Big)\\
	\nabla L(x,\lambda)=&0 \iff \left\{
	\begin{array}{l}
		\forall t, \big[x\,\overline{g}\big]'(x_t)=\dfrac{p_t-\lambda}{c_t}\\
		\displaystyle\sum_{t=1}^{T}x_t-N = 0 \\
	\end{array}
	\right.\\ 
\end{align*}
Focusing on $x\,\overline{g}$ as a univariate function, we get:
\begin{align*}
	x\,\overline{g}(x)=&\dfrac{p_t-\lambda}{c_t}\,x+b\\
	g\mbox{ is } \CC^0 \mbox{ in 0}&\mbox{, hence: } b=0\\
	\mbox{So: }x\,\overline{g}(x)=&\dfrac{p_t-\lambda}{c_t}\,x\\
\end{align*}
 We now introduce sufficient conditions to solve the Lagragian equations, when
 the price vector $p$ and consequently vector $c$ are constant. 
\begin{lemma}\label{lemCstPrice}
	 Under the assumption of Lemma \ref{lemXG}, if the price vector $p$ and consequently vector $c$ are constant, then the optimal strategy is $\hat{x}=(\frac{N}{T},\cdots,\frac{N}{T})$ and global maximum is  $\mbox{UB}_2=N\,\Big[p-c\,g\big(\frac{N}{T}\big)\Big]$
\end{lemma}
\begin{proof}
Lagrange equations yielded $x\,\overline{g}(x)=\dfrac{p-\lambda}{c}\,x$.\\ 
Function $g$ is strictly increasing and continuous on $\RR_+$, so it is injective. Hence, t being fixed, either $x_t=0$ or $x_t=\overline{g}^{-1}\big(\frac{p-\lambda}{c}\big)$, which is a unique value independent of t.\\
In addition, under the assumption of Lemma \ref{lemXG}, $x\cdot g$ is strictly convex, so $[x\cdot g]'$ is strictly increasing on $\RR_+$ and therefore injective.    
Let $x=(x_1,\cdots,x_T)$ satisfying Lagrange equations and $i<j$:
\[\big[x\,\overline{g}\big]'(x_i)=\big[x\,\overline{g}\big]'(x_j)=\dfrac{p-\lambda}{c}\]  
Therefore $x_i=x_j$. Since the null vector does not satisfy the constraint, we are left with $\forall, t,\; x_t=\frac{N}{T},\lambda=p-c\,g(\frac{N}{T})$, and $\mbox{UB}_2$ is obtained by direct computation of $U$.
\end{proof}
\noindent When price are not constant, we can still get an upper bound. Indeed, let $\overline{p}=\displaystyle\max_t{p_t}$ and $\underline{c}=\displaystyle\min_{t}{c_t}$. Then for all t: $u_t(x_t)\leq\Big[\overline{p}-\underline{c}\,\overline{g}(x_t)\Big]x_t$\\ 
By applying Lemma \ref{lemCstPrice} to the right hand side, we get:
\begin{lemma}\label{lemNonCst}
	Under Lemma \ref{lemXG} assumptions, with $\overline{p}=\displaystyle\max_t{p_t}$ and $\underline{c}=\displaystyle\min_{t}{c_t}$, the following inequality hold: 
	\[\mbox{UB}_2\leq N\,\Big[\overline{p}-\underline{c}\,g\big(\frac{N}{T}\big)\Big]=\overline{\mbox{UB}_2}\]
\end{lemma}
\begin{proof}
\begin{equation*}
	U(x)\leq\displaystyle\sum_{t=1}^{T}\Big[\overline{p}- \underline{c}\cdot \overline{g}(x_t)\Big]x_t
\end{equation*} 
We apply the Lemma \ref{lemCstPrice} to the right hand side problem. 
\end{proof}
We note, $\mbox{UB}_2$ and $\overline{\mbox{UB}_2}$ are equal when prices are constant. Hence, we refer to the latter in numerical experiments. 

%% file: introNE_V3.tex
In previous sections, we described techniques to obtain lower and upper bounds. We now study their numerical performances in terms of quality and CPU time. This section is organized as follows.\\In paragraph \ref{sectionData}, we detail the numerical experiment design for results reproducibility. We describe the machine characteristics, problem parameters selection and the price vector simulations. We also discuss the penalty function calibration process and its underlying motivations.\\Then in paragraph \ref{sectionNESmall}, we present our results for small and medium size instances. We start by introducing metrics, table notations and define instances size.\\
\indent In paragraph \ref{sectionNEDPExact}, we present the exact resolution via DP and discuss its applicability. 
While naive heuristics and ILS algorithm can be used straightforwardly, two-step approaches performance depends on their grain. So they require a proper setup which is discussed in the two following subsections. \\\indent In paragraph \ref{sectionNETSCoarse}, we describe the two-step approach based on coarse grain DP. We present its optimal coverage ratio and discuss its complexity as a function of the grain.
\\\indent  Then, in paragraph \ref{sectionNETSCont}, we compare our results to two-step approach based on continuous solutions.\\
\indent In paragraph \ref{sectionNEAggrSmall}, having fine tuned our two-step heuristics,  we present aggregated results for small and medium size instances, with both lower and upper bounds. We discuss quality and CPU time, in particular as stock prices fluctuate.\\
In paragraph \ref{sectionNELarge}, we move on to large instances for which exact resolution is not available. We shortly discuss time and memory limitations of our algorithms. We present quality and CPU time results providing a gap, although not tight, for the optimal value of the initial problem. (\ref{pbMain0}), when both lower and upper bounds are available. As for small and medium instances, we present representative examples to show the influence of stock price variation.

%% file: dataCalibration_V3.tex
We begin with the machine characteristics and softwares.\\\\
\textbf{PC characteristics:} numerical experiments were run using a PC with Intel Xeon(R) Silver 4114 at 2.20 Ghz , 2 sockets, 20 core, and 32 Gb of RAM. O/S is Linux Ubuntu 18.04 (Bionic Beaver) and c++ compiler gcc v9.3.0\\\\
\textbf{Commercial solvers:} LocalSolver 9.5 (v9.5, Linux64, build 20201030) and NLopt v2.6.2 
Then, we set up the parameters of the problem.\\\\
\textbf{Problem size} $(T, N)=(10^{a},10^{b}),\; a<b$, with $1\leq a\leq 3$ and $2\leq b\leq 9$ are labeled in the results tables. In particular, T divides N, as there is no specific interest to deal with odd blocks. When there is no ambiguity in the results tables, we further simplify notations by ignoring the base and write $(T,N)=(a,b)$. \\\\ 
\textbf{Minimum sale price} $q_t$: we assume $q_t=(1-\beta)p_t$, with $0<\beta<1$ constant. It means the minimum sale price for largest volumes (even a fire sale) is equal to a constant fraction of the current price. It can be interpreted as the intrinsic value of the company\footnote{An economical discussion about intrinsic value is beyond the scope of this paper, and not relevant for mathematical programming}. For numerical experiments, we set $\beta =0.9$. As $g(0)=0$ and g goes asymptotically to 1, this corresponds to a 90\% price floor. The closer the $\beta$ to 1, the stiffer is the penalty for selling more stocks, and the wider the penalty range (executed price $\in \big[(1-\beta)\,p_t;p_t\big]$). Consequently, a high value for $\beta$ leads to a significant difference in the objective value function, between a poor and a good selling strategy.\\\\ 		
We now shift our focus toward the asset price. We describe the stock price dynamic and the generation of price vectors.\\
\textbf{Stock price dynamic:} we obtain stock prices through simulations using a classical Geometric Brownian Motion stochastic process as in \cite{hull_2002}, starting at $p_0=100$, with moments $(\mu,\sigma)$:
		\[\dfrac{dS_t}{S_t}=\mu\,dt+\,\sigma\,dW_t\]
	By Ito's Lemma, we compute $d\left(\ln S_t\right)=\Big(\mu -\dfrac{\sigma^2}{2}\Big)dt+\sigma\,dW_t $
		
	Integrating over the time interval $[t;t+\Delta t]$, we get:
	\[ S_{t+\Delta t}=S_t \exp\left[ \Big(\mu -\dfrac{\sigma^2}{2}\Big)\Delta t+\sigma\,(W_{t+\Delta t}- W_t) \right]\]
	By stationariness of the Brownian motion:
	\[ S_{t+\Delta t}=S_t \exp\left[ \Big(\mu -\dfrac{\sigma^2}{2}\Big)\Delta t+\sigma\,W_{\Delta t} \right]\]
	Which yields a simpler formula for simulation purposes:
	\begin{equation}
		S_{t+\Delta t}=S_t \exp\left[ \Big(\mu -\dfrac{\sigma^2}{2}\Big)\Delta t+\sigma\,\sqrt{\Delta t}\,Z\right] ,\;\hbox{where}\;Z\sim \mathcal{N}(0,1)		
		\label{BMeq}
	\end{equation}\\\\
\textbf{Stock prices set up:} we simulate stock prices using the stochastic process described in equation (\ref{BMeq}), starting at $p_0=100$, with moments $(\mu,\sigma)$. We are interested in different trends, coupled with either a low or high volatility environment. Hence we set of moments:$\mu\in\{-0.05;0;+0.05\}$ and $\sigma\in\{0.10;0.25;0.70\}$, which leads to 9 combinations. We also perform numerical experiments using on a constant stock price $\forall i,\,p_i=100$. It can interpreted as a baseline experiment to compare algorithm performances.\\\\
One could legitimately object that high volatility diffusion processes are not realistic for our deterministic prices model. As mentioned in section \ref{sectionDefPb}, our optimization techniques apply to any price vector. From a mathematical programming standpoint, studying the bounds tightness for a wide range of standard responses is a topic of interest. Hence, we do not intend to draw conclusions for the financial application of problem (\ref{pbMain0}), based on unlikely instances, but rather to privilege a larger scope of application for our techniques, eventually outside of finance. To that end, it is important to validate our results when standard responses fluctuate significantly.\\\\
For each of the 9 sets $(\mu_i,\sigma_i)$, we run 10 simulations and take the average over $S_t$. It makes a smoother price path and better reflects moment characteristics.\\
Lastly, each of prices vector $V_i$ has the maximum $T_{max}=10^3$ cardinal considered for numerical experiments. When T is lower than than $T_{max}$, the corresponding price vector is drawn uniformly from $V_i$. For instance, when T=10, $\hat{V}=\{V_{100},V_{200},\cdots,V_{1000}\}$. One may notice there is a bias in the selection of the size $\lfloor\frac{T_{max}}{T}\rfloor$. We could indeed have arbitrarily chosen $\hat{V}=\{V_{j},V_{100+j},\cdots,V_{900+j}\}$, where $1\leq j\leq 100$. As T increases towards $T_{max}$, this bias fades out.\\      
Finally, we study the calibration process of penalty function g.\\\\ 
\textbf{Calibration of penalty function g:}
\begin{itemize}
	\item We start with the selected function prototypes G: $\begin{aligned}[t]
			\left\{
			\begin{array}{l}
				x \mapsto \displaystyle \frac{x}{1+x}\\
				or\; x\mapsto 1-\displaystyle \frac{2}{1+\sqrt{1+x}}\\
				or\; x\mapsto \displaystyle\frac{2}{\pi}\arctan(x)\\			
			\end{array}
			\right.
		\end{aligned}$\\
	\item G functions are $\mathcal{C}^2$ and bounded on $\mathbb{R}^+$. Indeed, $G(0)=0\;and\;\displaystyle\lim_{+\infty}G=1$
	\item We then define the penalty function $g:x\mapsto G(\eta\, x)$, where $\eta$ is constant and depends only on $G$. Given a level L , we define a threshold H, such that $g(L)=H$. Hence, $\eta$ is a scaling factor aimed to transpose the positive semi-line on the $[0;L]$ segment. Selected G functions are injective on $\mathbb{R}^+$, so we can compute $\eta=\dfrac{G^{-1}(H)}{N}$, where $G^{-1}$ is given respectively by:
		 $\begin{aligned}[t]
		 	\left\{
		 	\begin{array}{l}
		 		x \mapsto \displaystyle \frac{x}{1-x}\\
		 		or\; x\mapsto \displaystyle \tan\big(\frac{\pi}{2}\,x\big)\\	
		 		or\, x\mapsto \displaystyle \frac{4\,x}{(1-x)^2}		
		 	\end{array}
		 	\right.
		 \end{aligned}$\\   
	\item We lastly set threshold $H$. We considered different values for $H \in ]0;1[$. The closer to 1, the more discriminatory power for g. By discriminatory, we mean the distance from naive heuristics, such as the fire sale or the uniform sale, to the optimal, is maximum. For numerical experiments, we settled for $L=N$ and $H=0.99$. Underlying justifications and calibration tables are presented in Appendix \ref{sectionCalib}.\\ 
\end{itemize}

%% file: numExp_SmallMed_V3.tex
\noindent As described in section \ref{sectionData}, we simulated 10 price vectors, with different distributions moments.\\
For tables readability, we display metrics (i.e quality and CPU time) by default for constant prices and on average over the 9 simulated processes.
In the following result tables, CST refers to constant prices and AVG to  average. When results differ materially between price vectors, we mention it explicitly. \\
In results tables measuring bounds quality, figures are expressed in percentage and measure the relative difference to optimal coming from exact DP. For CPU time tables, time is measured in seconds. $\epsilon$ corresponds to the minimum numerical value in both cases, with its corresponding unit. Hence, ''$<\epsilon$ '' means either relative difference to the optimal is inferior to $0.01\%$ or  computation time is faster than  $0.01s$. Lastly, DNC means the algorithm either did not converge within allowed time of 10 minutes, or returned an memory error.\\
\textbf{Instances size}: for numerical experiments, we will consider the instance size as:
\begin{itemize}
	\item Small, if $(T,N)\in\big\{(1,2),(1,3),(1,4),(2,3),(2,4)\big\}$
	\item Medium, if $(T,N)\in\big\{(1,5),(1,6),(2,5),(3,5)\big\}$
	\item Large, if $(T,N)\in\big\{(1,7),(1,8),(1,9),(2,6),(2,7),(2,8),(2,9),(3,6)\big\}$
	\item Very large, when $(T,N)\in\big\{(3,7),(3,8),(3,9)\big\}$
\end{itemize}
We now present the exact resolution for small and medium instances.
\subsubsection{Exact resolution for small and medium instances :}\label{sectionNEDPExact}
we solve exactly problem (\ref{pbMain0}) for small and medium size instances using exact DP algorithm from section \ref{sectionDPExact}. CPU time is displayed in the table below.
\\\\\input{dpexact_cpu_tab1.tex}
\begin{itemize}
	\item Time complexity is $O(T\,N^2)$ as expected.
	\item For $(T,N)=(2,6)$, we only computed result in the CST case and it takes about $180\; 000\,s$, for each penalty function, in agreement with expected time complexity. 
	\item Time resolution does not depend on $p$ nor on the penalty function.
	It is expected as $(p_1,\cdots,p_T)$ and $ \left(g(1), \cdots, g(T)\right)$ are computed and stored ahead of resolution. 
	\item When $(T,N)\geq (2,5)$ (in a general sense), exact resolution takes from a few hours to a few days. It is hence not tractable for practical applications.  
\end{itemize}
Therefore, exact DP is not suitable for some medium size and large instances. Hence, we now introduce lower bounds results. \\
While ILS and naive heuristics are directly applicable, two-step approach depends on the grain, which controls both bucket and funnel size. We first present results for two-step approach and discuss optimal grain size.

\subsubsection{DP coarse grain and funnel}\label{sectionNETSCoarse}
As discussed in section \ref{sectionLB}, we run a DP with grain $P$ and then refine the solution with the same $\lambda. P$ size funnel. We computed results for $P=10$ to $N/10$, and $\lambda\in\{1,5\}$ (we tested different values of $\lambda\in[0;10]$ and settled for these as the most representative).
We first display the \textbf{optimal coverage ratio} (number of instances where the optimal is reached divided by total number of instances) as a function of the grain P. In the table below, optimal ratio is the first number expressed in percent. The second number, in parenthesis, corresponds to the total number of instances:
\\\\\input{ts1_tab1.tex}
\begin{itemize}
	\item A $5P$ funnel (i.e $\lambda=5,\lambda\ P= 50, 500, 5000$ etc.) provides a better coverage ratio for every P than $\lambda=1$.
	\item Maximum coverage ratio for a significant number of instances is reached for $P=100$ and funnel $\lambda P=500$. Optimal is reached for almost every small and medium size instances. 
	\item Quality is similar for CST and AVG. 
\end{itemize}
We display the CPU time table for that set up $\lambda P=500$: 
\\\\\input{ts1_cpu_tab1.tex}
\begin{itemize}
	\item CPU time is roughly similar for every penalty function.
	\item Resolution remains below a minute up to $(T,N)=(2,6)$
	\item It takes a slightly longer time for constant prices than for averaged batches. Indeed, for volatile process with peaks, optimal strategy consists in liquidating most of the block in the peaks area, leaving most other trading times with very few transactions. Thus, the effective number of time steps is lower.   
\end{itemize}
\textbf{Complexity: } minimizing time complexity in $P$ is also a very interesting topic. Numerically, with $T=10^a, N=10^b, P=10^c$, we set $\lambda=5$ and we apply the complexity results from section \ref{sectionLB}. Coarse grain DP is in $O(10^{a+2b-2c})$ and DP with bounds is in $O(10^{2(a+c+1)})$. Time complexity of the two-step method turns out to be in $O(10^{\max(a+2b-2c,2a+2c+2)})$. Its theoretical minimum is reached for $c=\frac{1}{4}(2b-a-2)$, with minimum time complexity in $O(10^{\frac{3a}{2}+b+1})$ (or $O(T^{\frac{3}{2}}\cdot N)$). If we restrict $c$ to the natural integers, we have to round it to the nearest integer, then $\overline{c}=\lfloor c+0.5\rfloor$, and achieved time complexity is in $O\left(10^{\max(a+2b-2\overline{c},2a+2\overline{c}+2)}\right)$. \\
Compared to exact DP in $O(10^{a+2b})$, we gain a factor $10^{b-\frac{a}{2}-1}$ (or $\frac{N}{10\sqrt{T}}$). Consequently, time improvement made thanks to the two-step method grow when N grows with respect to T. As we only require $P$ to be an integer, but not $c$, we recommend to let $c$ be in $\QQ$ to take full advantage to complexity gain from the two-step method and round $P=\lfloor 10^c\rfloor$.\\We now couple this techniques with the continuously relaxed problem.   

\subsubsection{Two-step approach based on continuous relaxation}\label{sectionNETSCont}
We mentioned in section \ref{sectionLB} that DP with bounds algorithm may also apply to continuous solution. Hence, we use a local maximum of problem (\ref{pbCont}) obtained using NLopt, with $l_t,u_t$ defined at the end of section \ref{sectionCGrain}.  We set the same parameter $(\lambda,\,P)$ as previously for comparison consistency. We first display the optimal ratio with P:
\\\\\input{ts2_tab1.tex}
\begin{itemize}
	\item Continuous relaxation with $\lambda P=500$ reached the optimal for every small and medium size instances.   
\end{itemize}
We display similarly the CPU time table for that set up $\lambda P=500$: 
\\\\\input{ts2_cpu_tab1.tex}
\begin{itemize}
	\item Two-step approach coupled with continuous relaxation revolves around the same CPU time as its coarse grain variation. Differences lie in the first step (NLopt heuristic vs. coarse grain DP), while DP with bounds complexity is relatively unchanged as is the funnel size. We however note that for $(T,N)= (1,6)$ or $(2,6)$, two-step with NLopt goes much faster than coarse grain counterpart. 
\end{itemize}
The two-step approaches, discrete or continuous, are now clearly defined. In the next section, we present aggregated results, for all our algorithms applied to small and medium size instances.
\subsubsection{Aggregated results for small and medium size instances}\label{sectionNEAggrSmall}
Having fine tuned two-step approaches, we can now compare our lower and upper bounds. The next two tables below present quality and CPU time. We introduce a few notations for table readability.\\
\textbf{Table notations :} \textbf{FS} refers to the naive fire sale heuristic, \textbf{US} to the uniform sale,  \textbf{TS1} stands for two-step with coarse grain DP, \textbf{TS2} relates to two-step with NLOpt heuristic, \textbf{ILS} corresponds to the discrete gradient described in section (\ref{sectionILS}) and initialized with uniform sale. Finally, \textbf{LS} refers to LocalSolver ran with approximately the same time limit as the best lower bound (capped to 10 minutes). \textbf{UB} corresponds to the upper bound $\overline{\mbox{UB}_2}$.\\Results profiles are similar for different penalty functions, hence we only display results for prototype $G=\frac{2}{\pi}arctan(x)$ for CST and AVG price batches, which are presented in the following two tables:
\\\\\input{recap_cst_tab2.tex}\\
\\\input{recap_avg_tab2.tex}
\begin{itemize}
	\item As expected, naive heuristic FS and US yields the worst lower bounds. FS remains around $20\%-25\%$ lower than the optimal, while US gets tighter when T increases. When T and N are of the same order of magnitude, there is enough time to liquidate the block, and the strategy becomes less relevant. Hence a simple uniform sale gets close to be optimal.
	\item LocalSolver returns a very good lower bound for every instance, but rather seldom the optimum. Its optimal coverage ratio is about 10\%  for both CST and AVG.
	\item TS1 and TS2 reached the optimum almost every time for small and medium size instances.
	\item UB is not tight even for small instances. It is stable in $N$, but becomes materially looser when $T$ increases.
	
	\item ILS algorithm performs better for CST than for AVG, especially as instance size grows. For AVG, it performs only marginally better than US. In addition, for AVG results differs significantly among instances.	
\end{itemize}
 We provide a representative example for $(T,N)=(3,5)$ in the following table:
\\\\\input{focus_ILS_tab1.tex}
\begin{itemize}
	\item ILS performs materially worse for high volatility instances with larger peaks rather than smoother ones. Those instances concentrate most of their liquidation around price peaks. We conclude that shifting adjacent time steps is less efficient for these profiles.
	\end{itemize}	
 Lastly, we display the equivalent CPU time tables:
\\\\\input{recap_cst_cpu_tab2.tex}\\
\\\\\input{recap_avg_cpu_tab2.tex}
\begin{itemize}
	\item For small and medium size instances, every algorithm converges relatively quickly (within a couple of minutes).
	\item Besides naive heuristics, ILS is the fastest algorithm (but does not necessarily yield a good result), followed by two-step approaches. 
	\item UB consists of a straightforward formula which is almost instantaneous for all $(T,N)$.
\end{itemize}
We covered the small and medium size instances and we are now interested in applying our techniques to large instances.

%% file: dpexact_cpu_tab1.tex
\begin{tabular}{|c|c|c|c|c|c|c|c|}
\hline
\multicolumn{2}{|c|}{CPU}&\multicolumn{2}{|c|}{$\frac{x}{1+x}$}& \multicolumn{2}{|c|}{$\frac{2}{\pi}\arctan(x)$}& \multicolumn{2}{|c|}{$1-\frac{2}{1+\sqrt{1+x}}$}\\
\hline
T&N&CST&AVG&CST&AVG&CST&AVG \\
\hline\hline
$10^1$&$10^2$&$<\epsilon$&$<\epsilon$&$<\epsilon$&$<\epsilon$&$<\epsilon$&$<\epsilon$             \\\hline
$10^1$&$10^3$&0.03&0.02&0.03&0.02&0.03&0.02             \\\hline
$10^1$&$10^4$&1.77&1.77&1.78&1.78&1.77&1.79             \\\hline
$10^1$&$10^5$&181&182&179&178&176&176    								 \\\hline
$10^2$&$10^3$&0.19&0.19&0.18&0.18&0.18&0.18             \\\hline
$10^2$&$10^4$&18.83&18.93&17.68&17.63&17.55&17.53       \\\hline
$10^2$&$10^5$&1829&1777&1786&1763&1752&1754             \\\hline
$10^3$&$10^5$&17 550&17 943&17 576&17 856&17 488&17 949            \\\hline
$10^1$&$10^6$&18 141&18 253&18 261&18 142&17 876&18 440            \\\hline
\end{tabular}

%% file: ts1_tab1.tex
\begin{tabular}{|c|c|c|c|c|c|c|c|c|}
\hline
$\lambda P=$&10&50&100&500&1000&5000&10000&50000    \\\hline
CST&27 (22)&95 (22)&23 (13)&100 (13)&14 (7)&100 (7)&0 (2)&100 (2)     \\\hline
AVG&42 (177)&97 (177)&49 (104)&99 (104)&44 (50)&98 (50)&33 (18)&100 (18)            \\\hline
TOTAL&41 (199)&97 (199)&46 (117)&99 (117)&40 (57)&98 (57)&30 (20)&100 (20)          \\\hline
\end{tabular}

%% file: ts1_cpu_tab1.tex
\begin{tabular}{|c|c|c|c|c|c|c|c|}
\hline
\multicolumn{2}{|c|}{CPU}&\multicolumn{2}{|c|}{$\frac{x}{1+x}$}& \multicolumn{2}{|c|}{$\frac{2}{\pi}\arctan(x)$}& \multicolumn{2}{|c|}{$1-\frac{2}{1+\sqrt{1+x}}$}\\
\hline
T&N&CST&AVG&CST&AVG&CST&AVG \\
\hline\hline
$10^1$&$10^2$&$<\epsilon$&$<\epsilon$&$<\epsilon$&$<\epsilon$&$<\epsilon$&$<\epsilon$             \\\hline
$10^1$&$10^3$&0.02&0.017&0.017&0.016&0.016&0.015        \\\hline
$10^1$&$10^4$&0.127&0.108&0.118&0.102&0.143&0.076       \\\hline
$10^1$&$10^5$&0.206&0.199&0.211&0.189&0.204&0.133       \\\hline
$10^2$&$10^3$&0.15&0.15&0.15&0.15&0.151&0.15            \\\hline
$10^2$&$10^4$&2.122&1.795&2.122&1.797&2.121&1.608       \\\hline
$10^2$&$10^5$&13.939&6.817&12.898&6.771&13.997&6.771    \\\hline
$10^3$&$10^5$&213.399&137.701&213.357&138.243&213.378&124.176         \\\hline
$10^1$&$10^6$&2.103&2.003&2.104&2.002&2.035&1.914       \\\hline
$10^2$&$10^6$&39.585&27.914&39.581&28.097&37.319&25.135 \\\hline
\end{tabular}

%% file: ts2_tab1.tex
\begin{tabular}{|c|c|c|c|c|c|c|}
\hline
$\lambda P=$&10&50&100&500&1000&5000            \\\hline
CST&50 (20)&65 (20)&73 (11)&100 (11)&100 (5)&100 (5)  \\\hline
AVG&56 (172)&81 (172)&90 (99)&100 (99)&100 (45)&100 (45)          \\\hline
TOTAL&56 (192)&79 (192)&88 (110)&100 (110)&100 (50)&100 (50)      \\\hline
\end{tabular}

%% file: ts2_cpu_tab1.tex
\begin{tabular}{|c|c|c|c|c|c|c|c|}
\hline
\multicolumn{2}{|c|}{CPU}&\multicolumn{2}{|c|}{$\frac{x}{1+x}$}& \multicolumn{2}{|c|}{$\frac{2}{\pi}\arctan(x)$}& \multicolumn{2}{|c|}{$1-\frac{2}{1+\sqrt{1+x}}$}\\
\hline
T&N&CST&AVG&CST&AVG&CST&AVG \\
\hline\hline
$10^1$&$10^2$&0.578&0.417&0.38&0.444&0.333&0.503        \\\hline
$10^1$&$10^3$&0.337&0.433&0.368&0.413&0.252&0.469       \\\hline
$10^1$&$10^4$&0.403&0.511&0.565&0.492&0.336&0.49        \\\hline
$10^1$&$10^5$&0.198&0.267&0.201&0.251&0.192&0.206       \\\hline
$10^2$&$10^3$&0.88&1.525&0.864&1.45&0.446&1.916         \\\hline
$10^2$&$10^4$&2.196&2.775&2.135&2.703&2.167&2.842       \\\hline
$10^2$&$10^5$&12.531&6.484&12.463&6.517&12.73&5.04      \\\hline
$10^3$&$10^5$&219.582&132.47&221.32&127.751&224.953&87.648            \\\hline
$10^1$&$10^6$&0.198&0.182&0.198&0.182&0.206&0.141       \\\hline
$10^2$&$10^6$&19.348&10.015&19.405&10.175&17.889&7.15   \\\hline
\end{tabular}

%% file: recap_cst_tab2.tex
\begin{tabular}{|c|c|c|c|c|c|c|c|c|}
\hline
\multicolumn{2}{|c|}{Quality}&\multicolumn{7}{|c|}{CST, $\frac{2}{\pi}arctan(x)$}\\
\hline
T&N&FS&US&ILS&LS&TS1&TS2&UB \\\hline
$10^1$&$10^2$&20.42&7.81&0.41&0&0&0&38.19 \\\hline
$10^1$&$10^3$&20.52&7.92&0.02&$<\epsilon$&0&0&38.02     \\\hline
$10^1$&$10^4$&20.52&7.92&$<\epsilon$&$<\epsilon$&0&0&38.02            \\\hline
$10^1$&$10^5$&20.52&7.92&$<\epsilon$&$<\epsilon$&0&0&38.02            \\\hline
$10^1$&$10^6$&20.52&7.92&$<\epsilon$&$<\epsilon$&0&0&38.02            \\\hline
$10^2$&$10^3$&25&2.12&1.04&$<\epsilon$&0&0&364.61       \\\hline
$10^2$&$10^4$&25&2.14&0.13&0.01&0&0&364.56              \\\hline
$10^2$&$10^5$&25.01&2.14&$<\epsilon$&0.01&0&0&364.56    \\\hline
$10^3$&$10^5$&25.49&0.23&0.12&0.10&0&0&558.72            \\\hline
\end{tabular}

%% file: recap_avg_tab2.tex
\begin{tabular}{|c|c|c|c|c|c|c|c|c|}
\hline
\multicolumn{2}{|c|}{Quality}&\multicolumn{7}{|c|}{AVG, $\frac{2}{\pi}arctan(x)$}\\
\hline
T&N&FS&US&ILS&LS&TS1&TS2&UB \\\hline
$10^1$&$10^2$&21.66&10.47&0.96&0&0&0&139.4              \\\hline
$10^1$&$10^3$&21.77&10.58&0.72&$<\epsilon$&0&0&139.07   \\\hline
$10^1$&$10^4$&21.77&10.59&0.71&$<\epsilon$&0&0&139.07   \\\hline
$10^1$&$10^5$&21.77&10.59&0.71&$<\epsilon$&0&0&139.07   \\\hline
$10^1$&$10^6$&21.77&10.59&0.71&$<\epsilon$&0&0&139.07   \\\hline
$10^2$&$10^3$&28.48&7.77&5.26&$<\epsilon$&0&0&419.26    \\\hline
$10^2$&$10^4$&28.49&7.8&4.04&0.01&0&0&419.13            \\\hline
$10^2$&$10^5$&28.50&7.8&3.97&0.01&0&0&419.13             \\\hline
$10^3$&$10^5$&30.49&5.98&7.39&0.14&0&0&576.05           \\\hline
\end{tabular}

%% file: focus_ILS_tab1.tex
\begin{tabular}{|c|c|c|}
\hline
Quality&\multicolumn{2}{|c|}{$\frac{2}{\pi}arctan(x)$}	\\\hline
\multicolumn{3}{|c|}{T=$10^3$, N=$10^5$}	\\\hline\hline
$\mu$&$\sigma$&ILS           \\\hline
-0.05  &0.10&2.13         \\\hline
-0.05  &0.25&3.46         \\\hline
-0.05  &0.70&8.76         \\\hline
0.00  &0.10&2.06        \\\hline
0.00  &0.25&5.02         \\\hline
0.00  &0.70&13.67         \\\hline
+0.05  &0.10&3.66         \\\hline
+0.05  &0.25&5.33         \\\hline
+0.05  &0.70&22.42         \\\hline
\end{tabular}

%% file: recap_cst_cpu_tab2.tex
\begin{tabular}{|c|c|c|c|c|c|c|c|c|}
\hline
\multicolumn{2}{|c|}{CPU}&\multicolumn{7}{|c|}{CST, $\frac{2}{\pi}arctan(x)$}\\
\hline
T&N&FS&US&ILS&LS&TS1&TS2&UB \\\hline
$10^1$&$10^2$&$<\epsilon$&$<\epsilon$&$<\epsilon$&10&$<\epsilon$&0.38&$<\epsilon$   \\\hline
$10^1$&$10^3$&$<\epsilon$&$<\epsilon$&$<\epsilon$&10&0.02&0.37&$<\epsilon$          \\\hline
$10^1$&$10^4$&$<\epsilon$&$<\epsilon$&$<\epsilon$&10&0.12&0.57&$<\epsilon$          \\\hline
$10^1$&$10^5$&$<\epsilon$&$<\epsilon$&$<\epsilon$&10&0.21&0.20&$<\epsilon$           \\\hline
$10^1$&$10^6$&$<\epsilon$&$<\epsilon$&$<\epsilon$&10&2.10&0.20&$<\epsilon$            \\\hline
$10^2$&$10^3$&$<\epsilon$&$<\epsilon$&0.01&10&0.15&0.86&$<\epsilon$   \\\hline
$10^2$&$10^4$&$<\epsilon$&$<\epsilon$&0.22&10&2.12&2.14&$<\epsilon$   \\\hline
$10^2$&$10^5$&$<\epsilon$&$<\epsilon$&0.73&600&12.90&12.46&$<\epsilon$ \\\hline
$10^3$&$10^5$&$<\epsilon$&$<\epsilon$&35.42&600&213.36&221.32&$<\epsilon$           \\\hline
\end{tabular}

%% file: recap_avg_cpu_tab2.tex
\begin{tabular}{|c|c|c|c|c|c|c|c|c|}
\hline
\multicolumn{2}{|c|}{CPU}&\multicolumn{7}{|c|}{AVG, $\frac{2}{\pi}arctan(x)$}\\
\hline
T&N&FS&US&ILS&LS&TS1&TS2&UB \\\hline
$10^1$&$10^2$&$<\epsilon$&$<\epsilon$&$<\epsilon$&10&$<\epsilon$&0.44&$<\epsilon$   \\\hline
$10^1$&$10^3$&$<\epsilon$&$<\epsilon$&$<\epsilon$&10&0.02&0.41&$<\epsilon$          \\\hline
$10^1$&$10^4$&$<\epsilon$&$<\epsilon$&$<\epsilon$&10&0.11&0.49&$<\epsilon$          \\\hline
$10^1$&$10^5$&$<\epsilon$&$<\epsilon$&$<\epsilon$&10&0.19&0.25&$<\epsilon$          \\\hline
$10^1$&$10^6$&$<\epsilon$&$<\epsilon$&$<\epsilon$&10&2.02&0.18&$<\epsilon$          \\\hline
$10^2$&$10^3$&$<\epsilon$&$<\epsilon$&$<\epsilon$&10&0.15&1.45&$<\epsilon$          \\\hline
$10^2$&$10^4$&$<\epsilon$&$<\epsilon$&0.07&10&1.8&2.7&$<\epsilon$     \\\hline
$10^2$&$10^5$&$<\epsilon$&$<\epsilon$&0.14&600&6.77&6.52&$<\epsilon$  \\\hline
$10^3$&$10^5$&$<\epsilon$&$<\epsilon$&2.92&600&138.24&127.75&$<\epsilon$            \\\hline
\end{tabular}

%% file: numExp_Large_V3.tex
At this scale, exact DP is not available. Moreover several algorithms do not converge for (very) large instances in the allowed time or are subject to memory constraint. We start by discussion these limitations.\\ 
\textbf{Memory limitations and potential improvements:} a double takes 8 bytes in the heap memory and space complexity of DP based algorithms is in $O(T\,N)$. A quick computation shows that, for our 32Gb RAM computer, the limit is $\frac{\ln(T\,N)}{\ln(10)}\leq 2^{32}\frac{\ln(2)}{\ln(10)}\approx 9.63<10$. Therefore, in our numerical experiments, we can't go any further than $(T,N)=(1,8),\;(2,7)$ or $(3,6)$ for two-step methods.\\ 
Indeed, within the DP with bounds algorithm, the search for the best solution occurs within a bounded funnel of size R, for which we showed space complexity is $O(T^2\,R)$. Hence, a T by N sparse matrix wastes too much memory, while a leaner data structure could save memory. The corresponding space gain is in the order of $\frac{T\,N}{T^2\,R}=\frac{N}{2\lambda\,P\,T}$. For our numerical experiments with (very) large instances, $N\gg T$, by a factor at least $10^3$ (and up to $10^6$). Therefore, this improved data structure would make sense. We leave it as a perspective in section \ref{sectionCCL}.
\textbf{Results tables presentation:} in the bound quality table presented below, the best lower bound for each instance is specified in the third column, against which quality is defined, as the relative value to the best known lower bound. 
\\\\\input{recap_cst_largeinstances_tab2.tex}\\\\\\
\input{recap_avg_largeInstances_tab2.tex}
\begin{itemize}
	\item In both constant and average cases:
	\begin{itemize}
		\item TS1 is the best lower bound, whenever available. TS2 is $\epsilon$ close to TS1, but not better. So contrary to intuition, a good continuous solution does not necessarily yield the best solution through neighborhood search.
		\item The fire sale is underperforming by about about 20-30\% for all $(T,N)$.
		\item Upper bound is not tight ranging from $40\%$ to seven fold, when compared to the best lower bound. Hence the interval $\big[\mbox{BEST LB; UB}\big]$ for the initial problem remains large, even in the constant case.
	\end{itemize}
	\item In the constant prices case:
	\begin{itemize}
		\item ILS is the best lower bound when two-step approach is not available. So ILS beats LocalSolver.
		\item However, all lower bounds are close, within a $0.1\%$ radius.
		\item Uniform sale is trailing by about $7\%$ when $T=10$. 	Its gap is stable in $N$ but decreases in $T$ and reaches less than $1\%$ for $T=10^3$. 
	\end{itemize}
		\item In the average case:
		\begin{itemize}
			\item When two-step are available, LS is close to TS1. The gap is less than $1\%$, but it seems to grow with $T$. 
			\item When two-step are not available, LS becomes the best lower bound. So LocalSolver beats ILS.  
			\item ILS trails the best bound by about about 1-10\%, then again the gap grows\\ with $T$.
			\item Similar to small and medium size instances, LS and ILS gap differ, sometime significantly, among instances. We discuss the results for a complete batch when $(T,N)=(3,6)$.  
		\end{itemize}
\end{itemize}
\input{focus_365_tab1.tex}\\\\
\begin{itemize}
	\item  In line with previous results, ILS accuracy (or lack thereof) decreases materially with volatility, and effect is more pronunced in a higher expected return environment.
	\item We observe the same pattern for LocalSolver, with the exception of $(\mu,\,\sigma)=(-0.05,0.25)$.
	\item ILS beats LocalSolver in the constant prices case only and is beaten in the average case. As mentioned previously, ILS is performing the worst for high volatility cases, because adjacent steps are not as efficient for highly volatile stock prices. We discuss potential improvements in section \ref{sectionCCL}.
\end{itemize}
Lastly, we display corresponding CPU tables:
\\\\\input{recap_cst_largeinstances_cpu_tab3.tex}\\\\\\
\input{recap_avg_largeInstances_cpu_tab2.tex}\\\\\\
\\\\\input{focus_365_cpu_tab1.tex}\\
\begin{itemize}
	\item We released the previous 10 minutes time cap, for very large instances, to compare final results.
	\item ILS is fast for large instances. However it grows slowly with $N$ but very rapidly, with (empirically in $T^3$ in the table).
	\item TS1 complexity results perform as expected in  $O(T^{\frac{3}{2}}\cdot N)$. TS1 CPU time remains tractable for large instances. However TS1 and TS2 are not available for very larges instances due to memory constraints.	
	\item TS2 is slightly faster than TS1, except for $(T,N)=(1,8),(2,6)$ where it is significantly faster.
	\item Upper bound is computed straightforwardly.
  	\item When $(T,N)=(3,6)$:
  	\begin{itemize}
  		\item ILS converges much faster when volatility increases as it gets quickly stuck in the first local maximum it returns.
  		\item TS1 CPU time also decreases with $\sigma$. However, on the contrary to ILS, DP based algorithm seems well suited for timeseries with elevated peaks which concentrate most of the sale.
  	\end{itemize} 
 \end{itemize}

%% file: recap_cst_largeinstances_tab2.tex
\begin{tabular}{|c|c|c|c|c|c|c|c|c|c|}
\hline
\multicolumn{3}{|c|}{Quality}&\multicolumn{7}{|c|}{CST, $\frac{2}{\pi}arctan(x)$}\\
\hline
T&N&BEST LB&FS&US&ILS&LS&TS1&TS2&UB       \\\hline
$10^1$&$10^7$&TS1 (5000)&20.52&7.92&$<\epsilon$&$<\epsilon$&0&$<\epsilon$&38.02     \\\hline
$10^1$&$10^8$&TS1 (5000)&20.52&7.92&$<\epsilon$&$<\epsilon$&0&$<\epsilon$&38.02      \\\hline
$10^1$&$10^9$&ILS&20.52&7.92&0&$<\epsilon$&DNC&DNC&38.02           \\\hline
$10^2$&$10^6$&TS1 (5000)&25.01&2.14&$<\epsilon$&0.01&0&$<\epsilon$&364.56 	 \\\hline
$10^2$&$10^7$&TS1 (5000)&25.01&2.14&$<\epsilon$&0.01&0&$<\epsilon$&364.56     \\\hline
$10^2$&$10^8$&ILS&25.01&2.14&0&$<\epsilon$&DNC&DNC&364.56             \\\hline
$10^2$&$10^9$&ILS&25.01&2.14&0&$<\epsilon$&DNC&DNC&364.56 \\\hline
$10^3$&$10^6$&TS1 (500)&25.49&0.23&0.01&0.07&0&$<\epsilon$&558.72   \\\hline
$10^3$&$10^7$&ILS&25.49&0.23&0&0.07&DNC&DNC&558.72             \\\hline
$10^3$&$10^8$&ILS&25.49&0.23&0&0.07&DNC&DNC&558.72            \\\hline
$10^3$&$10^9$&ILS&25.49&0.23&0&0.07&DNC&DNC&558.72            \\\hline
\end{tabular}

%% file: recap_avg_largeInstances_tab2.tex
\begin{tabular}{|c|c|c|c|c|c|c|c|c|c|}
\hline
\multicolumn{3}{|c|}{Quality}&\multicolumn{7}{|c|}{AVG, $\frac{2}{\pi}arctan(x)$}\\
\hline
T&N&BEST LB&FS&US&ILS&LS&TS1&TS2&UB       \\\hline
$10^1$&$10^7$&TS1 (5000)&21.77&10.59&0.71&$<\epsilon$&0&$<\epsilon$&139.07          \\\hline
$10^1$&$10^8$&TS1 (5000)&21.77&10.59&0.71&$<\epsilon$&0&$<\epsilon$&139.07          \\\hline
$10^1$&$10^9$&LS&21.77&10.59&0.71&0&DNC&DNC&139.07 \\\hline
$10^2$&$10^6$&TS1 (500)&28.49&7.8&3.97&$<\epsilon$&0&$<\epsilon$&419.13    \\\hline
$10^2$&$10^7$&TS1 (5000)&28.49&7.8&3.97&0.01&0&$<\epsilon$&419.13     \\\hline
$10^2$&$10^8$&LS&28.49&7.79&3.97&0&DNC&DNC&419.16       \\\hline
$10^2$&$10^9$&LS&28.43&7.72&4.33&0&DNC&DNC&443.43       \\\hline
$10^3$&$10^6$&TS1 (500)&30.35&7.87&7.19&0.29&0&$<\epsilon$&576.06     \\\hline
$10^3$&$10^7$&LS&30.31&7.82&7.11&0&DNC&DNC&576.49       \\\hline
$10^3$&$10^8$&LS&30.29&7.79&7.08&0&DNC&DNC&576.71       \\\hline
$10^3$&$10^9$&LS&30.31&7.81&7.11&0&DNC&DNC&576.52       \\\hline
\end{tabular}

%% file: focus_365_tab1.tex
\begin{tabular}{|c|c|c|c|c|}
\hline
Quality&\multicolumn{4}{|c|}{$\frac{2}{\pi}arctan(x)$}  \\\hline
\multicolumn{5}{|c|}{T=$10^3$, N=$10^6$} \\\hline
$\mu$&$\sigma$&BEST LB&ILS&LS             \\\hline\hline
\multicolumn{2}{|c|}{constant prices}&TS1(500)&0.01&0.07  \\\hline
-0.05&0.10&TS1(500)&2.08&0.04              \\\hline
-0.05&0.25&TS1(500)&3.45&1.23             \\\hline
-0.05&0.70&TS1(500)&8.75&0.31              \\\hline
0&0.10&TS1(500)&1.74&0.05    \\\hline
0&0.25&TS1(500)&5.01&0.08   \\\hline
0&0.70&TS1(500)&13.65&0.35   \\\hline
0.05&0.10&TS1(500)&2.26&0.09 \\\hline
0.05&0.25&TS1(500)&5.32&0.26              \\\hline
0.05&0.70&TS1(500)&22.42&0.25              \\\hline
\end{tabular}

%% file: recap_cst_largeinstances_cpu_tab3.tex
\begin{tabular}{|c|c|c|c|c|c|c|c|c|}
\hline
\multicolumn{2}{|c|}{CPU}&\multicolumn{7}{|c|}{CST, $\frac{2}{\pi}arctan(x)$}\\
\hline
T&N&FS&US&ILS&LS&TS1&TS2&UB              \\\hline
$10^1$&$10^7$&$<\epsilon$&$<\epsilon$&$<\epsilon$&600&20.92&19.11&$<\epsilon$   \\\hline
$10^1$&$10^8$&$<\epsilon$&$<\epsilon$&4&600&197.31&19.61&$<\epsilon$           \\\hline
$10^1$&$10^9$&$<\epsilon$&$<\epsilon$&36&600&DNC&DNC&$<\epsilon$     \\\hline
$10^2$&$10^6$&$<\epsilon$&$<\epsilon$&1.31&600&39.58&19.40&$<\epsilon$       \\\hline
$10^2$&$10^7$&$<\epsilon$&$<\epsilon$&1.99&600&1948.18&1918.30&$<\epsilon$       \\\hline
$10^2$&$10^8$&$<\epsilon$&$<\epsilon$&6&600&DNC&DNC&$<\epsilon$      \\\hline
$10^2$&$10^9$&$<\epsilon$&$<\epsilon$&38&600&DNC&DNC&$<\epsilon$     \\\hline
$10^3$&$10^6$&$<\epsilon$&$<\epsilon$&2220.02&3600&1466.14&1269.08&$<\epsilon$ \\\hline
$10^3$&$10^7$&$<\epsilon$&$<\epsilon$&7930.42&3600&DNC&DNC&$<\epsilon$             \\\hline
$10^3$&$10^8$&$<\epsilon$&$<\epsilon$&14719.71&10800&DNC&DNC&$<\epsilon$         \\\hline
$10^3$&$10^9$&$<\epsilon$&$<\epsilon$&25646.61&10800&DNC&DNC&$<\epsilon$         \\\hline
\end{tabular}

%% file: recap_avg_largeInstances_cpu_tab2.tex
\begin{tabular}{|c|c|c|c|c|c|c|c|c|}
\hline
\multicolumn{2}{|c|}{CPU}&\multicolumn{7}{|c|}{AVG, $\frac{2}{\pi}arctan(x)$}\\
\hline
T&N&FS&US&ILS&LS&TS1&TS2&UB              \\\hline
$10^1$&$10^7$&$<\epsilon$&$<\epsilon$&$<\epsilon$&600&19.21&5.7&$<\epsilon$         \\\hline
$10^1$&$10^8$&$<\epsilon$&$<\epsilon$&3.34&600&195.45&18.01&$<\epsilon$             \\\hline
$10^1$&$10^9$&$<\epsilon$&$<\epsilon$&38.33&60&DNC&DNC&$<\epsilon$    \\\hline
$10^2$&$10^6$&$<\epsilon$&$<\epsilon$&0.21&60&28.1&10.18&$<\epsilon$  \\\hline
$10^2$&$10^7$&$<\epsilon$&$<\epsilon$&0.29&600&1041.73&1011.34&$<\epsilon$          \\\hline
$10^2$&$10^8$&$<\epsilon$&$<\epsilon$&4.89&60&DNC&DNC&$<\epsilon$     \\\hline
$10^2$&$10^9$&$<\epsilon$&$<\epsilon$&39.34&60&DNC&DNC&$<\epsilon$    \\\hline
$10^3$&$10^6$&$<\epsilon$&$<\epsilon$&154.12&600&557.35&396.7&$<\epsilon$           \\\hline
$10^3$&$10^7$&$<\epsilon$&$<\epsilon$&257.41&600&DNC&DNC&$<\epsilon$  \\\hline
$10^3$&$10^8$&$<\epsilon$&$<\epsilon$&366.52&600&DNC&DNC&$<\epsilon$  \\\hline
$10^3$&$10^9$&$<\epsilon$&$<\epsilon$&648.6&600&DNC&DNC&$<\epsilon$   \\\hline
\end{tabular}

%% file: focus_365_cpu_tab1.tex
\begin{tabular}{|c|c|c|c|c|}
\hline
CPU &\multicolumn{4}{|c|}{$\frac{2}{\pi}arctan(x)$}  \\\hline
\multicolumn{5}{|c|}{T=$10^3$, N=$10^6$} \\\hline
$\mu$&$\sigma$&TS1(500)&ILS&LS            \\\hline\hline
\multicolumn{2}{|c|}{constant prices}&1466.14&2220.02&3600\\\hline
-0.05&0.1&408.06&52.62&600 \\\hline
-0.05&0.25&448.05&3.15&600 \\\hline
-0.05&0.7&340.44&1.04&600  \\\hline
0&0.1&729.33&229.31&600    \\\hline
0&0.25&657.45&8.58&600     \\\hline
0&0.7&294.03&1.59&600      \\\hline
0.05&0.1&882.91&1082.77&600\\\hline
0.05&0.25&656.1&6.99&600   \\\hline
0.05&0.7&599.75&1.07&600   \\\hline
\end{tabular}

%% file: ccl_V3.tex
We solved the non convex, integer, non linear mathematical program (\ref{pbMain0}), with a linear constraint, exactly for small instances using dynamic programming. We also found either the optimal or a very tight lower bound for medium sizes instances thanks to the two-step method based on hybrid DP (coarse grain or continuous relaxation coupled with DP with bounds). We derived its complexity and compared it to the exact DP. We provided different approaches to get tight lower bounds for medium size instances. Numerically, we beat LocalSolver in most cases. We also obtain an upper bound which is not tight\\
For most larges instances, our two-step method is available and we provided a tight lower bound which beats LocalSolver. Upper bound provides us with an interval for the optimal value of the initial problem which is not thin.\\ 
For some large and very large instances, where two-step method cannot be applied due to memory constraints, leaving the Iterated Local Search the only option available to us. We are beating LocalSolver only in the constant case, and losing to it in the non constant case. The upper bound provides is again with a wide interval for the optimal value.\\
While numerical experiments were necessarily performed on a few select penalty functions, our lower bound techniques apply, as discussed in section \ref{sectionDefPb}, to any real increasing function of $\NN$ and upper bound techniques to any $\CC^1$ real increasing function of $\RR_+$.\\
We now discuss the shortcomings of our approaches and the perspectives for future work. We begin with a technical improvement and then present methodological perspectives for future research.
\begin{itemize}
	\item \textbf{Memory management: } the efficient data structure suggested in section \ref{sectionNELarge} would enable to improve memory management and potentially gain an order of magnitude.
	\item \textbf{Upper bound quality:} our upper bound is not tight and gets wider when $T$ grows. A better upper bound would provide a tighter gap for the optimal value, especially for (very) large instances.
	\item \textbf{Iterated Local Search:} although ILS algorithm returns very good results in the constant case, it does not fare well when prices fluctuations are wild. While shifting adjacent time steps is not efficient in that case, one can shift $x_i,\, x_{i+k}$ for arbitrary k. An interesting question is how to choose the sequence of $k$'s to improve quality in reasonable CPU time. 
\end{itemize}

%% file: calibTables_VF.tex
Calibration factor $\eta$ is completely defined with $L$ and $H$, according to the calibration equation $G(\eta\,L)=H$. Sequence $y_t$ ($\sum_{i=1}^t x_i$) goes to N as t goes to T. Since, $\lim_{+\infty} g=1$, it is natural for g to get close to 1 as y goes to N. Hence $L=N$ seems a logical choice.\\
How close we get to 1 is precisely the role of threshold $H$. We considered different values for $H \in ]0;1[$. As mentioned previously, our goal is to preserve a gap between naive heuristics and optimal solution. A significant gap enables us to perceive more easily the quality of the different heuristics we introduced.
The lack thereof, on the other hand, yields a flatter landscape where it is more difficult, numerically, to exhibit the best strategies. We also want this gap to remain stable, or at least not to vanish, as T and N grow.\\     
Prices are taken constant for calibration purposes. We tested different value for $H$ but displayed only $H=0.75$, $H=0.99$, which are representative. Lastly we also presented $\eta=1$, as a baseline, which corresponds to the no calibration $(g\equiv G)$ case.\\
In the following calibration tables, FS corresponds to the fire sale strategy, US the uniform sale. Columns corresponds to prototype function G. Figures are expressed in $\%$ and measure relative difference to optimal (exact DP) solution:\\
\begin{enumerate}
	\item\input{CalibTable3.tex}
	\item\input{CalibTable5.tex}
	\item\input{CalibTable0.tex}
\end{enumerate}
\textbf{Conclusion:} 
\begin{itemize}
	\item  For functions $x \mapsto \frac{x}{1+x}$ and $x\mapsto \frac{2}{\pi}\arctan(x)$, $\eta_{0.99}$ calibration has the most discriminatory power.
	\item For function $x\mapsto 1-\frac{2}{1+\sqrt{1+x}}$, $\eta_{0.75}$ is slightly better for small instances and then vanishes quickly when $T\geq 2$. $\eta_{0.99}$ is more stable and offers more discriminatory power for medium size instances.
	\item For scaling factor $\eta_{0.99}$, FS distance to optimal remains stable as T and N grow. In the US case, it is stable in N, but decreases, albeit slower compared to other calibrations, in T.
	\item Lastly, we notice that functions $x \mapsto \frac{x}{1+x}$ and $x\mapsto \frac{2}{\pi}\arctan(x)$ have roughly the same convergence speed (it is related to the similarity of their first derivative' expression), while $x\mapsto 1-\frac{2}{1+\sqrt{1+x}}$ exhibits lower figures as US distance to optimal is smaller than $1\%$, across all instances.
\end{itemize}

%% file: CalibTable3.tex
\begin{tabular}{|c|c|c|c|c|c|c|c|}
\hline
\multicolumn{2}{|c|}{\boldmath{$\eta_{0.75}$}}&\multicolumn{2}{|c|}{$\frac{x}{1+x}$}& \multicolumn{2}{|c|}{$\frac{2}{\pi}\arctan(x)$}& \multicolumn{2}{|c|}{$1-\frac{2}{1+\sqrt{1+x}}$}
\\
\hline
T&N&FS&US&FS&US&FS&US
\\
\hline\hline
$10^1$&$10^2$&33.44&0.84&37.85&0.52&23.58&1.43
\\
\hline
$10^1$&$10^3$&33.45&0.84&37.86&0.53&23.58&1.43
\\
\hline
$10^1$&$10^4$&33.45&0.84&37.86&0.53&23.58&1.43
\\
\hline
$10^1$&$10^5$&33.45&0.84&37.86&0.53&23.58&1.43
\\
\hline
$10^1$&$10^6$&33.45&0.84&37.86&0.53&23.58&1.43
\\
\hline
$10^2$&$10^3$&36.65&0.09&40.9&0.05&26.77&0.24
\\
\hline
$10^2$&$10^4$&36.65&0.09&40.9&0.05&26.77&0.24
\\
\hline
$10^2$&$10^5$&36.65&0.09&40.9&0.05&26.77&0.24
\\
\hline
$10^3$&$10^5$&36.97&0.01&41.19&0.01&27.1&0.03
\\
\hline
\end{tabular}

%% file: CalibTable5.tex
\begin{tabular}{|c|c|c|c|c|c|c|c|}
\hline
\multicolumn{2}{|c|}{\boldmath{$\eta_{0.99}$}}&\multicolumn{2}{|c|}{$\frac{x}{1+x}$}& \multicolumn{2}{|c|}{$\frac{2}{\pi}\arctan(x)$}& \multicolumn{2}{|c|}{$1-\frac{2}{1+\sqrt{1+x}}$}
\\
\hline
T&N&FS&US&FS&US&FS&US
\\
\hline\hline
$10^1$&$10^2$&18.27&6.04&20.42&7.81&5.48&0.94
\\
\hline
$10^1$&$10^3$&18.38&6.17&20.52&7.92&5.51&0.98
\\
\hline
$10^1$&$10^4$&18.38&6.17&20.52&7.92&5.51&0.98
\\
\hline
$10^1$&$10^5$&18.38&6.17&20.52&7.92&5.51&0.98
\\
\hline
$10^1$&$10^6$&18.38&6.17&20.52&7.92&5.51&0.98
\\
\hline
$10^2$&$10^3$&22.63&1.97&25.00&2.12&7.00&0.53
\\
\hline
$10^2$&$10^4$&22.65&2.00&25.00&2.14&7.08&0.62
\\
\hline
$10^2$&$10^5$&22.65&2.00&25.01&2.14&7.08&0.62
\\
\hline
$10^3$&$10^5$&23.11&0.24&25.49&0.23&7.28&0.18
\\
\hline
\end{tabular}

%% file: CalibTable0.tex
\begin{tabular}{|c|c|c|c|c|c|c|c|}
\hline
\multicolumn{2}{|c|}{\boldmath{$\eta =1$}}&\multicolumn{2}{|c|}{$\frac{x}{1+x}$}& \multicolumn{2}{|c|}{$\frac{2}{\pi}\arctan(x)$}& \multicolumn{2}{|c|}{$1-\frac{2}{1+\sqrt{1+x}}$}
\\
\hline
T&N&FS&US&FS&US&FS&US
\\
\hline\hline
$10^1$&$10^2$&18.16&6.03&15.25&6.41&23.91&1.98
\\
\hline
$10^1$&$10^3$&3.44&1.79&2.53&1.46&18.43&2.59
\\
\hline
$10^1$&$10^4$&0.45&0.28&0.32&0.21&9.47&1.6
\\
\hline
$10^1$&$10^5$&0.05&0.04&0.04&0.02&3.68&0.66
\\
\hline
$10^1$&$10^6$&0.01&$<\epsilon$&$<\epsilon$&$<\epsilon$&1.25&0.23
\\
\hline
$10^2$&$10^3$&4.61&1.18&3.38&1.08&22.06&0.96
\\
\hline
$10^2$&$10^4$&0.68&0.3&0.47&0.23&11.9&0.89
\\
\hline
$10^2$&$10^5$&0.09&0.05&0.06&0.04&4.79&0.45
\\
\hline
$10^3$&$10^5$&0.09&0.03&0.06&0.02&4.92&0.14
\\
\hline
\end{tabular}

%% file: articleDPMain_V3.bbl
\begin{thebibliography}{10}

\bibitem{alfo_2008}
{\sc A.~Alfonsi, A.~Fruth, and A.~Schied}, {\em Constrained portfolio
  liquidation in a limit order book model}, Banach Center Publ, 83 (2008),
  pp.~9--25.

\bibitem{alfo_2010}
{\sc A.~Alfonsi, A.~Fruth, and A.~Schied}, {\em Optimal execution strategies in
  limit order books with general shape functions}, Quantitative Finance, 10
  (2010), pp.~143--157.

\bibitem{almg_2000}
{\sc R.~Almgren and N.~Chriss}, {\em Optimal execution of portfolio
  transactions}, Journal of Risk, 3 (2000).

\bibitem{almg_2003}
{\sc R.~F. Almgren}, {\em Optimal execution with nonlinear impact functions and
  trading-enhanced risk}, Applied Mathematical Finance, 10 (2003), pp.~1--18.

\bibitem{barle_1994}
{\sc G.~Barles}, {\em Solutions de viscosité des équations de
  Hamilton-Jacobi}, Springer, 1994.

\bibitem{bell_1957}
{\sc R.~Bellman}, {\em Dynamic Programming}, Princeton University Press,
  Princeton, NJ, USA, 1~ed., 1957.

\bibitem{localsolver_2020}
{\sc T.~Benoist, F.~Gardi, J.~Darlay, and R.~Megel}, {\em Localsolver}, 2020.
\newblock https://www.localsolver.com/home.html.

\bibitem{bert_1998}
{\sc D.~Bertsimas and A.~Lo}, {\em Optimal control of execution costs}, Journal
  of Financial Markets, 1 (1998).

\bibitem{bill_2012}
{\sc A.~Billionnet, S.~Elloumi, and A.~Lambert}, {\em Extending the qcr method
  to general mixed-integer programs}, Mathematical programming, 131 (2012),
  pp.~381--401.

\bibitem{boyd_2017}
{\sc S.~Boyd, E.~Busseti, S.~Diamond, R.~N. Kahn, K.~Koh, P.~Nystrup, and
  J.~Speth}, {\em Multi-period trading via convex optimization}, arXiv preprint
  arXiv:1705.00109,  (2017).

\bibitem{cplex_2018}
{\sc I.~I. CPLEX}, {\em V12.6.2 : User’s manual for cplex}, International
  Business Machines Corporation,  (2018).

\bibitem{dann_1977}
{\sc L.~Dann, D.~Mayers, and R.~Raab}, {\em Trading rules, large blocks and the
  speed of price adjustment}, Journal of Financial Economics, 4 (1977).

\bibitem{duran_1986}
{\sc M.~A. Duran and I.~E. Grossmann}, {\em An outer-approximation algorithm
  for a class of mixed-integer nonlinear programs}, Mathematical programming,
  36 (1986), pp.~307--339.

\bibitem{flou_1995}
{\sc C.~A. Floudas and V.~Visweswaran}, {\em Quadratic optimization}, in
  Handbook of global optimization, Kluwer Academic Publishers, Dordrecht, 1995,
  pp.~217--269.

\bibitem{garey_1979}
{\sc M.~R. Garey and D.~S. Johnson}, {\em Computers and Intractability; A Guide
  to the Theory of NP-Completeness}, WH Freeman \& Co., 1979.

\bibitem{gath_2010}
{\sc J.~Gatheral}, {\em No-dynamic-arbitrage and market impact}, Quantitative
  Finance, 10 (2010).

\bibitem{gath_2012}
{\sc J.~Gatheral and A.~Schied}, {\em Dynamical models of market impact and
  algorithms for order execution}, SSRN Electronic Journal,  (2012).

\bibitem{gemm_1996}
{\sc G.~Gemmill}, {\em Transparency and liquidity: A study of block trades on
  the london stock exchange under different publication rules}, The Journal of
  Finance, 51 (1996), pp.~1765--1790.

\bibitem{geof_1972}
{\sc A.~M. Geoffrion}, {\em Generalized benders decomposition}, Journal of
  optimization theory and applications, 10 (1972), pp.~237--260.

\bibitem{gams_2021}
{\sc G.~S. GmbH}, {\em Gams global library}, 2021.
\newblock http://www.gamsworld.org/global/globallib.htm.

\bibitem{gupta_1985}
{\sc O.~K. Gupta and A.~Ravindran}, {\em Branch and bound experiments in convex
  nonlinear integer programming}, Management science, 31 (1985),
  pp.~1533--1546.

\bibitem{guth_1965}
{\sc H.~G. Guthmann and A.~J. Bakay}, {\em The market impact of the sale of
  large blocks of stock}, The Journal of Finance, 20 (1965).

\bibitem{hiri_1995}
{\sc J.-B. Hiriart-Urruty}, {\em Conditions for global optimality}, in Handbook
  of global optimization, Kluwer Academic Publishers, Dordrecht, 1995,
  pp.~1--26.

\bibitem{horst_1995}
{\sc R.~Horst and P.~M. Pardalos}.

\bibitem{hull_2002}
{\sc J.~C. Hull}, {\em Options, futures and other derivatives}, Pearson
  Prentice-Hall, Upper Saddle River, NJ, USA, 5th~ed., 2002.

\bibitem{john_2021}
{\sc S.~G. Johnson}, {\em The nlopt nonlinear-optimization package}, 2021.
\newblock http://github.com/stevengj/nlopt.

\bibitem{keim_1996}
{\sc D.~B. Keim and A.~Madhavan}, {\em The upstairs market for large-block
  transactions: Analysis and measurement of price effects}, Review of Financial
  Studies, 9 (1996).

\bibitem{khar_2010}
{\sc I.~Kharroubi and H.~Pham}, {\em Optimal portfolio liquidation with
  execution cost and risk}, SIAM Journal on Financial Mathematics, 1 (2010),
  pp.~897--931.

\bibitem{madh_1997}
{\sc A.~Madhavan and M.~Cheng}, {\em In search of liquidity: Block trades in
  the upstairs and downstairs markets}, Review of Financial Studies, 10 (1997).

\bibitem{mish_2017}
{\sc R.~Mishra}, {\em Optimal portfolio liquidation in dark pool}, Master's
  thesis, Indian Statistical Institute, Kolkata, India, 07 2017.

\bibitem{obiz_2005}
{\sc A.~A. Obizhaeva and J.~Wang}, {\em Optimal trading strategy and
  supply/demand dynamics, working paper}, SSRN Electronic Journal,  (2005).

\bibitem{obiz_2013}
{\sc A.~A. Obizhaeva and J.~Wang}, {\em Optimal trading strategy and
  supply/demand dynamics}, Journal of Financial Markets, 16 (2013), pp.~1--32.

\bibitem{quad_2015}
{\sc D.~Quadri and E.~Soutil}, {\em Reformulation and solution approach for
  non-separable integer quadratic programs}, Journal of the Operational
  Research Society, 66 (2015), pp.~1270--1280.

\bibitem{ques_1992}
{\sc I.~Quesada and I.~E. Grossmann}, {\em An lp/nlp based branch and bound
  algorithm for convex minlp optimization problems}, Computers \& chemical
  engineering, 16 (1992), pp.~937--947.

\bibitem{sepp_1990}
{\sc D.~Seppi}, {\em Equilibrium block trading and asymmetric information}, The
  Journal of Finance, 45 (1990), pp.~73--94.

\bibitem{seydel_2009}
{\sc R.~C. Seydel}, {\em Existence and uniqueness of viscosity solutions for
  qvi associated with impulse control of jump-diffusions}, Stochastic Processes
  and their Applications, 119 (2009).

\bibitem{svan_2002}
{\sc K.~Svanberg}, {\em A class of globally convergent optimization methods
  based on conservative convex separable approximations}, SIAM journal on
  optimization, 12 (2002), pp.~555--573.

\bibitem{vath_2006}
{\sc V.~L. Vath, M.~Mnif, and H.~Pham}, {\em A model of optimal portfolio
  selection under liquidity risk and price impact}, Finance and Stochastics, 11
  (2007).

\end{thebibliography}
